\documentclass[12pt,a4paper]{article}
\usepackage[DIV=13]{typearea}
\usepackage[onehalfspacing]{setspace}

\def\spacingset#1{\renewcommand{\baselinestretch}%
	{#1}\small\normalsize} \spacingset{1}
\usepackage[english]{babel}
\usepackage[ansinew]{inputenc}
\usepackage{amsmath,amsfonts, amssymb, amsthm,bm,array,dsfont,graphicx,natbib}
\usepackage[colorlinks,citecolor=black,linkcolor=black,urlcolor=black]{hyperref}
\usepackage[hang,tight]{subfigure}
\usepackage{authblk}
\usepackage{framed}
\usepackage{xcolor}
\usepackage{bbm}
\usepackage{comment}

\colorlet{shadecolor}{gray!25}

\newcommand{\monthword}[1]{\ifcase#1\or Jan\or Feb\or M\"ar\or Apr\or Mai\or Jun\or Jul\or Aug\or Sep\or Okt\or Nov\or Dez\fi}
\newcommand{\leadingzero}[1]{\ifnum #1<10 0\the#1\else\the#1\fi}             
 \theoremstyle{plain}
 \newtheorem{assumption}{Assumption}
 \newtheorem{corollary}{Corollary}
\newtheorem{Theorem}{Theorem}
\newcommand{\todayI}{\the\year"~\leadingzero{\month}"~\leadingzero{\day}}    	
\newcommand{\todayII}{\the\year\leadingzero{\month}\leadingzero{\day}}       	
\newcommand{\todayIII}{\leadingzero{\day}/\leadingzero{\month}/\the\year}    	
\newcommand{\todayIV}{\leadingzero{\day}.\leadingzero{\month}.\the\year}     	
\newcommand{\todayV}{\the\day.\the\month.\the\year}                          	
\newcommand{\todayVI}{\the\day.~\monthword{\the\month}. \the\year}           	
\newcommand{\todayVII}{\leadingzero{\day}.~\monthword{\the\month}. \the\year}	
\newcommand{\todayVIII}{\monthword{\the\month}. \the\year}										

\newcommand{\mA}{\bm A}
\newcommand{\va}{\bm a}
\newcommand{\mB}{\bm B}

\newcommand{\mD}{\bm D}
\newcommand{\vd}{\bm d}

\newcommand{\ve}{\bm e}

\newcommand{\vf}{\bm f}

\newcommand{\mH}{\bm H}

\newcommand{\mI}{\bm I}

\newcommand{\mM}{\bm M}

\newcommand{\mP}{\bm P}

\newcommand{\mQ}{\bm Q}

\newcommand{\mR}{\bm R}

\newcommand{\mS}{\bm S}

\newcommand{\mT}{\bm T}

\newcommand{\vu}{\bm u}

\newcommand{\mX}{\bm X}

\newcommand{\vy}{\bm y}
\newcommand{\mZ}{\bm Z}
\newcommand{\vz}{\bm z}

\newcommand{\valpha}{\bm \alpha}

\newcommand{\vzeta}{\bm \zeta}

\newcommand{\vtheta}{\bm \theta}

\newcommand{\vxi}{\bm \xi}

\newcommand{\vchi}{\bm \chi}

\newcommand{\vvarphi}{\bm \varphi}

\newcommand{\mLambda}{\bm \varLambda}

\newcommand{\mSigma}{\bm \varSigma}

\newcommand{\mPsi}{\bm \varPsi}

\makeatletter
\newcommand{\rd}{\@ifnextchar^{\DIfF}{\DIfF^{}}}
\def\DIfF^#1{%
   \mathop{\mathrm{\mathstrut d}}%
   \nolimits^{#1}\gobblespace}
\def\gobblespace{\futurelet\diffarg\opspace}
\def\opspace{%
   \let\DiffSpace\!%
   \ifx\diffarg(%
   \let\DiffSpace\relax
   \else
   \ifx\diffarg[%
   \let\DiffSpace\relax
   \else
   \ifx\diffarg\{%
   \let\DiffSpace\relax
   \fi\fi\fi\DiffSpace}

\newcommand{\diag}{\operatorname{diag}}

\newcommand{\norm}[1]{\left\lVert #1 \right\rVert}

\renewcommand{\lim}[1]{\underset{#1}{\operatorname{lim}} \;}
\newcommand{\plim}[1]{\underset{#1}{\operatorname{plim}} \;}

\renewcommand{\min}[1]{\underset{#1}{min} \;}

\newcommand{\bvec}{\begin{pmatrix}}
	\newcommand{\evec}{\end{pmatrix}}
\newcommand{\bmat}{\begin{bmatrix}}
	\newcommand{\emat}{\end{bmatrix}}
\allowdisplaybreaks

\title{Macroeconomic Forecasting with Fractional Factor Models}
\author[1,2]{Tobias Hartl\footnote{Corresponding author. E-Mail: tobias1.hartl@ur.de\\
The author thanks Federico Carlini, Manfred Deistler, Christoph Rust, Rolf Tschernig, Enzo Weber, Roland Weigand,  and the participants at the Long Memory Conference 2018 in Aalborg, at the Workshop on Statistics and Econometrics 2018 in Passau, at the Econometric Society Europe Meeting 2018 in Cologne, at the Annual Meeting of the German Statistical Society 2018 in Linz, and at the International Conference on Computational and Financial Econometrics 2018 in Pisa for many valuable comments. Support through the projects TS283/1-1 and WE4847/4-1 financed by the German Research Foundation (DFG) is gratefully acknowledged.}}

\affil[1]{University of Regensburg, 93053 Regensburg, Germany}
\affil[2]{Institute for Employment Research (IAB), 90478 Nuremberg, Germany}
\date{May 2020}

\begin{document}
	\spacingset{1.45}
\maketitle

\thispagestyle{empty}
\setcounter{page}{0}
\paragraph{\bf Abstract.}
We combine high-dimensional factor models with fractional integration methods and derive models where
nonstationary, potentially cointegrated data of different persistence is modelled as a function of common fractionally integrated factors. 
A two-stage estimator, that combines principal components and the Kalman filter, is proposed. The forecast performance is studied for a high-dimensional US macroeconomic data set, where we find that benefits from the fractional factor models can be substantial, as they outperform univariate autoregressions, principal components, and the factor-augmented error-correction model.

\paragraph{\bf Keywords.}
Fractional integration, state space model, principal components, long memory, Kalman filter

\paragraph{\bf JEL-Classification.}

C32, C38, C51, C53

\newpage

\section{Introduction}
At least since the seminal work of \cite{ForHalLi2000} and \cite{StoWat2002}, factor models have become a popular tool for forecasting macroeconomic dynamics, as they handle covariation in the cross-section efficiently by condensing it to a typically small number of common latent factors. Regardless their applicability to large data sets, the major drawback of standard factor models is an inefficient use of longitudinal information: in contrast to e.g.\ VARMA models, the vast majority of factor models requires stationarity. Consequently, key features of macroeconomic data, such as nonstationary trends and cointegration, are not captured adequately by standard factor models but rather differenced away. Over-differencing of latent processes poses an additional risk, since model selection criteria and model specification tests for the number of factors are likely to miss these components, as eigenvalues corresponding to over-differenced series converge to zero. 

\noindent
A more flexible setup is suggested by a young strand of the factor model literature that adds unit roots to the model \citep[see, e.g.,][]{PenPon2006, Eic2009, ChaMilPa2009, BanMarMa2014, BanMarMa2016, BarLipLu2016}. But these models come with the drawback of requiring a priori assumptions about the degree of persistence, and typically the series under study are assumed to be $I(1)$. This makes an endogenous treatment of the (unknown) long-run dynamic characteristics of observable time series impossible. Statistical inference about the degree of persistence of an observable variable is then limited to prior unit root testing, ignoring the non-standard behavior of many economic series that are fractionally integrated. Misspecifying the integration orders of the observable variables may bias the factor estimates, can yield wrong inference about the number of common factors, and is likely to deteriorate the forecast performance. 

\noindent
To address these problems, semiparametric methods that are robust to fractional integration have been proposed by \cite{LucVer2015} for a single fractionally integrated factor and by \cite{Erg2017} for pervasive fractionally integrated nuisance. Allowing for a wide range of persistence and an endogenous treatment of integration orders, \cite{HarWei2018} derive a parametric fractionally integrated factor model and apply it to realized covariance matrices.

\noindent
In macroeconomics, fractionally integrated factor models have not played a role so far, although there is comprehensive evidence for long memory and fractional cointegration in the data \citep[cf.\ e.g.][]{HasWol1995, Bai1996, GilRob1997, TscWebWe2013}. 

\noindent
Tackling this issue, this paper aims to provide insights on whether fractional integration techniques have merit at least for a relevant fraction of the numerous and heterogeneous macroeconomic variables typically under study. 
By elaborating fractionally integrated factor models, we construct setups where cross-sectional covariation in the data in levels is driven by fractionally integrated latent factors that may impose cointegration relations. In detail, we propose three different factor models that generalize the aforementioned factor models to fractionally integrated processes. The first model introduces ARFIMA processes in the nonstationary factor model setup of \cite{BarLipLu2016}, while the second model distinguishes between purely fractionally integrated factors that impose cointegration relations and $I(0)$ factors that model common short-run behavior of the data. Finally, our third model generalizes the pre-differencing of the data for standard $I(0)$ factor models by taking fractional differences.

\noindent
As standard factor models they are applicable to high-dimensional data, but bear several advantages: the fractional factor models allow for a joint modelling of data of different persistence, do not require prior assumptions about the degree of persistence of the data but treat the integration orders endogenously, they capture cointegration via the common fractionally integrated factors and are more robust to over-differencing. 

\noindent
For the estimation of the latent factors we introduce a two-stage estimator, where initial factor estimates are obtained via principal components, until the model is cast in state space form such that the Kalman filter and smoother is applicable. For the latter to be computationally feasible, we use ARMA approximations for fractionally integrated processes as suggested in \cite{HarWei2018a}. Estimation of the unknown model parameters and the latent factors is then carried out jointly via an expectation maximization algorithm. 

\noindent 
 In a pseudo out-of-sample forecast experiment for a high-dimensional US macroeconomic data set of \cite{McCNg2015}, we study the forecast accuracy of the fractional factor models. We provide a guided choice among the different models by considering the forecast performance for $112$ macroeconomic variables. Finally, we find comprehensive evidence that adequately  combining fractional integration techniques and factor models can improve forecasts substantially compared to standard factor models and other benchmarks.

\noindent
The remaining paper is organized as follows. Section \ref{Chapter:2} details the construction of fractional factor models. The two-stage estimator for the factors and model parameters is discussed in section \ref{Chapter:3}. Section \ref{Ch:5} compares the forecast performance of the fractional factor models to different benchmarks in a pseudo out-of-sample forecast experiment, until section \ref{Ch:6} concludes.

\section{Fractional Factor Models} \label{Chapter:2}
To begin with, consider the general form of a high-dimensional factor model for possibly fractionally integrated data
\begin{align}\label{eq:factor_general}
	\vy_t = f(\vchi_t) + \vu_t, \qquad t=1,...,T,
\end{align}
where $\vy_t=\left(y_{1,t},...,y_{N,t}\right)' $ is an $N$-dimensional observable time series with entries $y_{i,t} \sim I(d^*_i)$ that are integrated of order $d_i^*$,  $d^*_i \in \mathbb{R}_{\geq 0}$. An integration order $d_i^*$ implies that the fractional difference of a series is $I(0)$, i.e.\ $\Delta^{d_i^*}y_{i,t}\sim I(0)$, $i=1,...,N$. The vector $\vchi_t$ is $r$-dimensional and accounts for common short- and long-run dynamics among the $\vy_t$, and $\vu_t= (u_{1, t},...,u_{N, t})'$ holds the $N$ idiosyncratic errors and has a  diagonal variance. 

\noindent
The fractional difference operator $\Delta^d$ is defined as 
\begin{align}\label{diff}
\Delta^{d} &= (1-L)^{d} = \sum_{j = 0}^{\infty}\pi_{j}(d)L^{j},  \qquad
\pi_{j}(d) = 
\begin{cases}
\frac{j-d-1}{j}\pi_{j-1}(d) &  j = 1, 2, ..., \\ 
1										&   j = 0,
\end{cases} 
\end{align}
and a $+$-subscript amounts to a truncation of an operator at $t \leq 0$, e.g.\ for an arbitrary stochastic process $z_t$, $\Delta^d_+ z_t = \sum_{j=0}^{t-1}\pi_j(d) L^j z_t$ \citep[see e.g.][]{Joh2008}. 
For $d \in \mathbb{N}_0$ fractionally integrated processes nest the standard integer integrated specifications (e.g.\ $I(0)$, $I(1)$, and $I(2)$ processes), whereas $d \in \mathbb{R}_{\geq 0}$ adds flexibility to the weighting of past shocks. Throughout the paper, we adopt the type II definition of fractional integration \citep{MarRob1999} that assumes zero starting values for all fractional processes, and, as a consequence, allows for a seamless treatment of the asymptotically stationary ($d < 1/2$) and the nonstationary ($d \geq 1/2$) case. Due to the type II definition the inverse fractional difference  $\Delta_+^{-d}z_t$ exists.

\noindent
Standard factor models as those considered in \citet{ForHalLi2000}, \citet{BaiNg2002}, and \cite{StoWat2002},
are special cases of \eqref{eq:factor_general}. They extract $r$ common factors of a data set in first and second differences, implying that the series in $\vy_t$ are $I(1)$ and $I(2)$. The common factors in $f(\vchi_t)$ then correspond to the common trends in the Granger representation theorem for cointegrated data \citep[see][]{BarLipLu2016}. 

\noindent 
To give an intuition on how fractional integration affects the long-run properties of a time series we note the following. 
For positive $d$ the autocovariance function of an $I(d)$ process decays at a hyperbolic rate, implying that a shock has a persistent impact on the $I(d)$ process, and a greater $d$ implies a more persistent impact of a shock. While an $I(d=1)$ process is an unweighted sum of past shocks, an $I(d)$ process in general can be interpreted as a weighted sum of past shocks, where the weights depend on $d$ via \eqref{diff}. Furthermore, if a linear combination of a vector $I(d)$ process exists that is integrated of order $b < d$, then the series are cointegrated. Cointegration implies common (fractionally) integrated trends, which our models capture via $f(\vchi_t)$. 
For a discussion of cointegration relations in a fractionally integrated factor model setup we refer to \cite{HarWei2018}.

\noindent
We introduce three different fractionally integrated factor models in the next sections that are nested in \eqref{eq:factor_general} and differ in the functional relation between $\vchi_t$ and $\vy_t$. Section \ref{Ch:2} generalizes nonstationary factor models \citep[cf.\ e.g.][]{BarLipLu2016} to fractionally integrated processes. In section \ref{Ch:3} we distinguish between fractionally integrated factors, that account for long-run co-movements in $\vy_t$, and $I(0)$ factors, that allow common short-run dynamics. Finally, section \ref{Ch:4} generalizes the pre-differencing of standard factor models to fractional differencing.

\subsection{Dynamic Fractional Factor Model}\label{Ch:2}
Consider a simple multivariate unobserved components model
\begin{align}\label{eq:1}
\vy_{t} = \bm{\Lambda} \vf_{t} + \vu_{t}, \qquad t = 1, ..., T,
\end{align}
where $f(\vchi_t)=\mLambda \vf_t$ in \eqref{eq:factor_general}, $\vf_t = (f_{1,t}, ..., f_{r,t})'$ holds the $r$ common factors, $\mLambda$ is a $N \times r$ matrix of factor loadings that is assumed to have full column rank, and the errors 
$\vu_t$ account for idiosyncratic dynamics. The latent factors are assumed to follow $r$ fractionally integrated autoregressive processes
\begin{align}\label{eq:2}
	B_j(L)
	\Delta_+^{d_j}
	f_{j,t} = {\zeta}_{j,t}, \qquad j = 1,...,r,
\end{align}
where ${B}_j(L)= 1-\sum_{k =1}^{p}{B}_{j,k}L^{k}$ is a stable lag polynomial. For the pervasive shocks that drive $\vf_t$  we assume $(\zeta_{1,t},...,\zeta_{r,t})'= \bm{\zeta}_{t}\sim \mathrm{NID}(\bm{0}, \mQ)$, where $\mQ$ is diagonal. A matrix formulation of \eqref{eq:2} follows directly by defining $\vd = (d_1,...,d_r)'$, the lag polynomials $\mD(\vd) = \mathrm{diag}(\Delta^{d_1}_+,...,\Delta^{d_r}_+)$ and $\mB(L) = \mathrm{diag}(B_1(L),...,B_r(L))$, such that $\mB(L)\mD(\vd)\vf_t = \vzeta_t$. 

\noindent
The errors ${u}_{i,t}$ are assumed to be mutually independent and are allowed to be autocorrelated
\begin{align}\label{eq:3}
	\rho_{i}(L)u_{i,t} = \xi_{i,t}, \qquad \xi_{i,t} \sim \mathrm{NID}(0, \sigma_{\xi_i}^2), \qquad i = 1, ..., N,
\end{align} 
where $\rho_{i}(L)=\sum_{k=0}^{p_i}\rho_{i,k}L^k$ is a stationary autoregressive lag polynomial. 

\noindent
As a consequence, the model may explain various degrees of common persistence that characterize the data by common components with long memory. For $d_1 =...=d_r =0$, the model nests the approximate dynamic factor model of \cite{StoWat2002}, while $d_j \in \{0, 1\}$, $j=1,...,r$, yields a nonstationary dynamic factor model with $I(1)$ factors as considered in e.g.\ \cite{BarLipLu2016}. Therefore, the model can be interpreted as a fractional generalization that neither requires prior differencing of the data, nor a prior assumptions about the integration orders.

\subsection{Dynamic Orthogonal Fractional Components} \label{Ch:3}
A more parsimonious factor model is proposed by \cite{HarWei2018}. Their model distinguishes between $r_1$ purely fractional factors $\vf_t^{(1)}=(f_{1, t}^{(1)},...,f_{r_1, t}^{(1)})'$, that establish cointegration relations among the $\vy_t$, and $r_2$ stationary autoregressive components $\vf_t^{(2)}=(f_{1, t}^{(2)},...,f_{r_2, t}^{(2)})'$, that account for common short-run behavior. We consider a slightly more general modification that allows for autocorrelated idiosyncratic errors. The general framework for the dynamic orthogonal fractional components model is then given by 
\begin{align}\label{eq:31}
\vy_{t} &= \bmat
\bm{\Lambda}^{(1)}&\bm{\Lambda}^{(2)}
\emat
\begin{pmatrix}
\vf_t^{(1)} \\ \vf_t^{(2)}
\end{pmatrix}
+ \vu_{t}, &&t=1,...,T\\
\Delta^{d_{j}}_+ f^{(1)}_{j,t} &= \zeta^{(1)}_{j,t}, &&j = 1,...,r_1, \label{DOFC:f1}\\
B^{(2)}_k(L)f^{(2)}_{k,t} &= \zeta^{(2)}_{k,t}, &&k = 1,...,r_2, \label{DOFC:f2}\\
\rho_i(L)u_{i,t}&=\xi_{i,t}, &&i=1,...,N, \label{DOFC:u}
\end{align}
for all $t = 1,...,T$ and $r=r_1 + r_2 \leq N$. The $N$ idiosyncratic shocks $\bm{\xi}_t = (\xi_{1,t},...,\xi_{N,t})'$ are assumed to follow independent Gaussian white noise processes $\bm{\xi}_t \sim \mathrm{NID}(\bm{0}, \bm{H})$. For the pervasive shocks $\vzeta_t^{(1)} = (\zeta_{1, t}^{(1)}, ..., \zeta_{r_1, t}^{(1)})'$, $\vzeta_t^{(2)} = (\zeta_{1, t}^{(2)}, ..., \zeta_{r_2, t}^{(2)})'$ we assume $\mathrm{vec}(\bm{\zeta}^{(1)}_{t}, \bm{\zeta}^{(2)}_t) \sim \mathrm{NID}(0, \mQ)$ where $\mQ$ is diagonal. In addition, we assume that the errors $\vu_t$ are independent of the common components $\vf_t$. 

\noindent
Define the polynomials $\mB^{(2)}(L) = \diag (B_1^{(2)}(L),...,B_{r_2}^{(2)}(L))$, $\mD^{(1)}(\vd) = \diag (\Delta_+^{d_1},...,\Delta_+^{d_{r_1}})$. Then, the model can be shown to be nested in the setup of section \ref{Ch:2} for $\vf_t=\mathrm{vec}(\vf_t^{(1)}, \vf_t^{(2)})$, $\mB(L) = \mathrm{diag}(\bm{I}, \mB^{(2)}(L))$, and $\mD(\vd)=\mathrm{diag}(\mD^{(1)}(\vd), \bm{I})$. In terms of \eqref{eq:factor_general} the model specifies $f(\vchi_t)=\mLambda^{(1)} \vf_t^{(1)} + \mLambda^{(2)} \vf_t^{(2)}$. 

\noindent
Note that the NID assumption on $\vzeta_t$ together with $\mQ$ diagonal yields $r$ orthogonal factors $\vf_t$. This common feature of many unobserved components models, which also applies to the models in sections \ref{Ch:2} and \ref{Ch:4}, reduces estimation uncertainty of the loadings and makes the framework very attractive for forecasting. Since $\vu_t$, $\vzeta_t$ are assumed to be independent, any correlation among the variables in $\vy_{t}$ stems from the common long- and short-run components $\vf^{(1)}_{t}$ and $\vf^{(2)}_{t}$. 

\subsection{Dynamic Factor Model in Fractional Differences} \label{Ch:4}
A third model that completes our toolbox of fractionally integrated factor models takes fractional differences of the observable variables to arrive at a short memory model, where all components are at most $I(0)$. Hence, we contrast our two models from sections \ref{Ch:2} and \ref{Ch:3} with an additional approach that excludes fractional integration from the factors. 
For this purpose we define
\begin{align}\label{sfd:1}
\Delta_+^{d_i^*}y_{i,t}&=\bm{\Lambda}_{i}\vf_{t}+{\xi}_{i,t}, &&t = 1,...., T, \hspace{0.5 cm} i = 1,..., N, \\
B_j(L)f_{j,t} &=  \zeta_{j,t},  &&j = 1, ..., r. \label{sfd:2}
\end{align}
As before, $\vy_t=(y_{1,t},...,y_{N,t})'$ are the observable variables, $\bm{\Lambda} = [\bm{\Lambda}_{1}', ..., \bm{\Lambda}_{N}']'$ holds the factor loadings, and $\vf_{t}=(f_{1,t},...,f_{r, t})'$ contains the $r$ latent factors. In the notation of \eqref{eq:factor_general} this implies $f(\vchi_t) = \mD(-\vd^*)\mLambda \vf_t$ and $\vu_t = \mD(-\vd^*)\vxi_t$ with $\vxi_t = (\xi_{1, t},..., \xi_{N, t})'$, and $\vd^* = (d_1^*,...,d_N^*)'$.

\noindent
By defining $\mB(L) = \mathrm{diag}({B}_{1}(L), ..., {B}_{r}(L))$ as in sections \ref{Ch:2} and \ref{Ch:3} the factors $\vf_t$ can be written as a diagonal VAR process $\mB(L)\vf_t = \vzeta_t$, where $\vzeta_t = (\zeta_{1, t}, ..., \zeta_{r, t})'$. The idiosyncratic and pervasive shocks are assumed to be orthogonal and to follow independent Gaussian white noise processes $\bm{\xi}_t \sim \mathrm{NID}(\bm{0}, \mH)$ and $\bm{\zeta}_t \sim \mathrm{NID}(\bm{0}, \mQ)$. 

\noindent
By taking fractional differences prior to estimating a factor model, our approach generalizes the pre-differencing of standard factor models to the fractional domain. In fractional differences, our model is an approximate dynamic factor model and, therefore, it nests the model of \cite{StoWat2002} for $d_1^*,...,d_N^* \in \mathbb{N}_0$.

\noindent
Taking fractional differences of order $d_i^{*}$ ensures for each $\Delta_+^{d_i^*}y_{i,t}$ that the common and idiosyncratic components are at most $I(0)$. Note that fractional differences are less sensitive to over-differencing compared to integer differences, since the method ensures that the fractional difference of the most persistent factor that loads on $y_{i,t}$ is $I(0)$ $\forall i =1,...,N$.  
\section{Estimation}\label{Chapter:3}
In this section we discuss both, the estimation of the latent factors $f(\vchi_t)$ in \eqref{eq:factor_general} for the three different factor models proposed in sections \ref{Ch:2} to \ref{Ch:4}, and the estimation of the unknown model parameters. The expectation-maximization (EM) algorithm is a natural choice for the estimation of parametric factor models \citep[cf.\ e.g.][]{JunKoo2015} and has been derived for fractionally integrated factor models in \citet[appendix B]{HarWei2018a}. In the E-step, the latent factors are estimated given a set of parameters via the Kalman filter. The M-step then updates the parameter vector by maximizing the likelihood function given the factor estimates from the E-step. Therefore, the EM-algorithm allows for a joint estimation of factors and model parameters. 

\noindent
Since the EM-algorithm is a parametric estimator, it requires starting values for the unknown model parameters in sections \ref{Ch:2} to \ref{Ch:4}. We tackle this problem by proposing a two-stage estimator. The first stage is described in section \ref{sec2.1}. We estimate the latent factors via the nonparametric method of principal components (PC) and propose estimators for the unknown model parameters. We include a consistency proof for the PC estimator for fractionally integrated factors with integration orders in $\mathbb{R}_{\geq 0}$, since consistency of the PC estimator has so far only been shown in more restrictive settings. 

\noindent
The second stage is considered in section \ref{sec2.2}. We derive an approximate state space formulation for each of the factor models in sections \ref{Ch:2} to \ref{Ch:4}, so that the Kalman filter can be applied to estimate the latent factors. Finally, we discuss the joint estimation of the model parameters and the latent factors via the EM algorithm.

\subsection{First Stage: Principal Components}\label{sec2.1}
Sufficient conditions for a consistent estimation of $f(\vchi_t)$ in \eqref{eq:factor_general} via PC were derived in \cite{BaiNg2002} for stationary processes, in \cite{Bai2004} for $I(1)$ common components and in \cite{BaiNg2004} for $\vy_t \sim I(1)$, where nonstationarity may also stem from the idiosyncratic components. For a single fractionally integrated factor and fractional integration orders in $[0, 1]$ \cite{LucVer2015} have shown that the methods of \cite{BaiNg2004} are also applicable. We generalize their results to non-negative integration orders and multiple fractionally integrated factors by showing consistency of the PC estimator for $f(\vchi_t)$ in \eqref{eq:factor_general}.

\noindent
Since PC are estimated via an eigendecomposition of $\mathrm{Var}(\vy_t)$, the applicability of the PC estimator depends crucially on the stability of the variance. For $\mathrm{max}(d_i^*)<0.5$ all $\vy_t$ are asymptotically stationary, and consequently the variance of $\vy_t$ converges as $t \to \infty$. Therefore, the PC estimator satisfies the assumptions of \cite{BaiNg2002}, where assumption A postulates boundedness of $\mathrm{plim}_{T \to \infty} T^{-1} \sum_{t=1}^T \vf_t \vf_t' = \mSigma_f < \infty$. 

\noindent
For $\mathrm{max}(d_i^*)\geq 0.5$ assumption A of \cite{BaiNg2002} is violated. Nonetheless, under a suitable scaling the PC estimator is still consistent. We report an updated set of assumptions for consistency of the PC estimator for nonstationary data in appendix \ref{sec:A1}. 
Following \cite{BaiNg2002} and \cite{Bai2004}, for $d_1 = ... = d_r$ we show that there exists a matrix $\mH$ such that the factors $ \vf_t$ are estimated consistently up to a rotation by PC
\begin{align*}
		\frac{1}{T}\sum_{t=1}^{T}||\hat{\vf}_t - \mH' \vf_t||^2 \xrightarrow{p} 0.
\end{align*}
Expressions for $\hat{\vf}_t$ and $\mH$, together with a detailed proof, are given in appendix \ref{sec:A1}.

\noindent
Whenever there is at least one $d_j \neq d_k$,  $j, k = 1,...,r$, direct estimation of all common fractionally integrated factors is not feasible, since, depending on the scaling of the PC, either the contribution of the least persistent factors to the covariance of $\vy_t$ converges to zero, or the contribution of the most persistent factors diverges. In this case, one needs to separate $\vy_t$ into blocks of equal persistence. Starting with the most persistent block, latent factors are estimated via PC and projected out. The adjusted variables are then added to the next block of $\vy_t$, and the procedure repeats, until a stationary set of variables is obtained. 

\noindent
Having established consistency of PC for the estimation of the latent fractionally integrated factors, we turn to the estimation of the dynamic parameters for the common factors. Since the dynamic properties differ among the three frameworks discussed in section \ref{Chapter:2}, we consider them separately in the following.

\paragraph{ARFI factors}
The common components of the model in section \ref{Ch:2} are assumed to follow $r$ independent autoregressive fractionally integrated processes. Therefore, we rotate the PC estimates via the method of \cite{MatTsa2011} to obtain dynamic orthogonal components. The parameters in \eqref{eq:2} are estimated by maximizing the likelihood function for a multivariate fractionally integrated process \citep[see][]{Nie2004} that is given by
\begin{align}\label{eq:4}
l(\vd, \mB, \mQ)= - \frac{T}{2}\mathrm{log}\vert \mQ \vert -\frac{1}{2}\sum_{t=1}^{T}\left(\mB(L)\mD(\vd)\vf_{t}\right)'\mQ^{-1}\left(\mB(L)\mD(\vd)\vf_{t}\right),
\end{align}
with $\mB = \mathrm{vec}(\mB_{1}, ..., \mB_{p})= \mathrm{vec}(\mB(L))$, and $\mB(L)$, $\mD(\vd)$ as defined in section \ref{Ch:2}. 
Plugging in the first-order condition
$\bm{\mQ} = \bm{\mQ}(\vd, \mB) = T^{-1}\sum_{t=1}^{T}\left(\mB(L)\mD(\vd)\vf_{t}\right)\left(\mB(L)\mD(\vd)\vf_{t}\right)'$,
and dropping the constant terms gives
$
l^{*}(\vd, \mB)=-\frac{T}{2}\mathrm{log}\vert\mQ(\vd, \mB)\vert,
$ which we maximize to get estimates for the unknown parameters $d_{1}, ..., d_{r}$, $\mB_{1}, ..., \mB_{p}$. 
For some data sets the assumption of orthogonal factors may be violated. Then, the diagonal assumption on $\mB(L)$ can be dropped, which does not affect the identification of the fractional factor VAR but increases the number of unknown parameters in \eqref{eq:4}.
Factor loadings $\mLambda$ in \eqref{eq:1} are estimated via ordinary least squares (OLS).

\paragraph{FI and AR factors}
To derive an estimator for the dynamic parameters of the 
model in section \ref{Ch:3}, we first need to distinguish between the space spanned by the purely fractional factors and the stationary autoregressive components. We identify the two factor subspaces of ${\vf}_t^{(1)}$ and ${\vf}_t^{(2)}$ up to a rotation by estimating the fractional cointegration subspace and its orthogonal complement via the semiparametric method of \cite{CheHur2006}, who use eigenvectors of an averaged periodogram matrix of the first $m$ Fourier frequencies to estimate the fractional cointegration subspace. Finally, orthogonal series within the fractional and non-fractional factors are obtained by applying the decorrelation method of \cite{MatTsa2011}. 
The resulting fractional and non-fractional factor estimates are denoted as $\hat{\vf}_{t}^{(1)}$ and $\hat{\vf}_{t}^{(2)}$ respectively.

\noindent
Given the factor estimates $\hat{\vf}_t^{(1)}$ and $\hat{\vf}_t^{(2)}$ together with the observable variables $\vy_t$, we estimate the factor loadings $\bm{\Lambda}$ in \eqref{eq:31} and the AR coefficients of \eqref{DOFC:f2} via OLS. Estimates for the fractional integration orders of the common components in \eqref{DOFC:f1} are obtained by maximizing the likelihood of the $r_1$ ARFIMA(0, $d_j$, 0) processes, $j=1,...,r_1$. 
\paragraph{AR factors}
Due to the stationary representation of the model in section \ref{Ch:4} the PC estimator of \cite{BaiNg2002} is directly applicable. The factors are again decorrelated by means of dynamic orthogonal components of \cite{MatTsa2011}. For a discussion of the consequences when a diagonal representation of the common factors is not feasible we refer to the ARFI case. The dynamic coefficients for the $r$ common factors in \eqref{sfd:2} together with their factor loadings in \eqref{sfd:1} are estimated via OLS.  
\paragraph{AR errors}
An estimate for the idiosyncratic errors is obtained via 
	$\hat{\vu}_t =	\vy_t - \hat{\mLambda}\hat{\vf}_t$.
Since the errors are assumed to follow $N$ independent autoregressive processes, the AR parameters are estimated via OLS. 
\subsection{Second Stage: Kalman Filter and Smoother}\label{sec2.2}
The second stage of our estimator combines factor estimation for a given set of parameters via the Kalman filter and smoother together with parameter optimization via maximum likelihood (ML) in an EM algorithm. For the Kalman filter to be applicable, the different components of our fractional factor models are cast in state space form. 
Note that for a given sample size $T$ a finite state space representation of a type II fractionally integrated process exists but requires a state vector of dimension $T-1$, as \eqref{diff} shows. Since the Kalman filter sequentially inverts the $(T-1) \times (T-1)$ autocovariance matrix for each factor, a full representation of a fractionally integrated process can be very costly from a computational perspective, in particular for long time series. 
Therefore, section \ref{Ch:3.2.1} discusses finite approximations that resemble the dynamic properties of fractionally integrated processes well and are computationally feasible. Section \ref{Ch:3.2.2} derives the state space representation and section \ref{Ch:3.2.3} considers parameter estimation. 
\subsubsection{Approximations for Fractionally Integrated Processes}\label{Ch:3.2.1}
The literature has considered a variety of approximations for long memory processes: \citet[section 4.2]{Pal2007} suggests truncated AR approximations, whereas \cite{ChaPal1998} study truncated MA approximations. In a simulation study, \cite{HarWei2018a} find that small ARMA($v, w$) models with $v, w \in \{3, 4\}$ outperform pure AR and MA approximations even if a high number of lags enter the latter models. In addition, the ML estimator for the integration order is found to be more precise when an ARMA approximation is used. As their simulation studies show, the ML estimates for an approximate representation of a fractionally integrated process converge to the ML estimates of the exact state space representation as $T \to \infty$. For the latter, consistency is proven in \cite{HarTscWeb2019}.

\noindent
 Following the suggestions of \cite{HarWei2018a}, an ARMA($4, 4$) process is used to approximate the purely fractional factors of section \ref{Ch:3}. For ARFIMA processes, whose dynamic properties stem not only from the fractional differencing operator, the approximation quality of ARMA processes is not clear. Therefore, we use pure AR($5$) processes to resemble the properties of the fractional differencing operator in the ARFI-case of section \ref{Ch:2}. For an arbitrary integration paramter $b$ the approximations are given by
\begin{align*}
\Delta_+^{b}\overset{a}{=} \left[\frac{a(L, b)}{m(L, b)}\right]_+=\left[ \frac{1 - a_{1}(b)L - ... - a_{v}(b)L^{v} }{ 1 + m_{1}(b)L + ... + m_{w}(b)L^{w}} \right]_{+},
\end{align*}
where $m_{k}(b)$ are the MA coefficients, $k = 1, ..., w$, and $a_{l}(b)$ are the AR parameters, $l = 1,...,v$, $(v=4, w=4)$ for purely fractional factors as in \eqref{DOFC:f1}, and $(v=5, w=0)$ for ARFI-factors as in \eqref{eq:2}. 

\noindent
The ARMA parameters are chosen beforehand for a given sample size $T$ and fractional integration order $b$ by minimizing the distance between the generic process $x_t = \Delta_+^{-b}z_t=\sum_{j=0}^{t-1}\pi_j(-b)z_{t-j}$ and its approximation $\tilde{x}_t = [m(L, b)a(L, b)^{-1}]_+z_t=\sum_{j=0}^{t-1}\tilde{\psi}_j(-b) z_{t-j}$, $z_t \sim NID(0, 1)$, over all $t=1,...,T$, where $\tilde{\psi}_j(-b)$ is the $j$-th coefficient of the ARMA Wold representation, and $\pi_j(-b)$ is its counterpart from \eqref{diff}. We use the mean squared error over $t=1,...,T$ as the distance measure
$
MSE_{T}^{b} = \frac{1}{T} \sum_{t=1}^{T}\sum_{j=0}^{t-1}\left(\tilde{{\psi}}_j(-b) - {\pi}_j(-b)\right)^2.
$

\noindent
For a given sample size $T$ and integration order $b$, we collect the ARMA coefficients in a $(v+w)$-vector $\vvarphi_T(b)=(a_1(b),...,a_v(b), m_1(b),...,m_w(b))'$. The ARMA coefficient estimates are then defined via $\hat{\vvarphi}_T(b) = \arg \min \vvarphi MSE_T^b$. Following \cite{HarWei2018a}, for a given $T$ optimization is carried out for each value on a grid for $b$. The ARMA coefficients are smoothed using cubic regression splines, such that a continuous, differentiable function $\vvarphi_T(b)$ in $b$ is obtained. The technical details and several simulation studies are contained in \cite{HarWei2018a}.

\noindent
With a smooth function $\vvarphi_T(b)$ in $b$ at hand, parameter optimization for the models in sections \ref{Ch:2} and \ref{Ch:3} can be conducted over the low-dimensional vector of fractional integration orders $\vd$, which keeps the dimension of the parameter vector within the optimization procedure manageable and independent of the length of the ARMA approximations. 

\subsubsection{State Space Representation of Fractional Factor Models} \label{Ch:3.2.2}
With these approximations at hand, we can turn to the state space representation of our fractional factor models. A general representation of a state space model is given by
\begin{align}\label{eq:5}
\tilde{\vy}_{t} &= \mZ\bm{\alpha}_{t} + \bm{\xi}_{t} 
, 
&&\bm{\alpha}_{t+1} = \mT\bm{\alpha}_{t}+\mR\bm{\zeta}_{t+1}, 
\end{align}
where $\va_{t|T}=\mathrm{E}(\bm{\alpha}_{t}|\tilde{\vy}_1,...,\tilde{\vy}_{T}, \bm{\theta})$, $\mP_{t|T}=\mathrm{Var}(\bm{\alpha}_t|\tilde{\vy}_1,...,\tilde{\vy}_T, \bm{\theta})$. The covariance matrices of the disturbances $\mQ = \mathrm{Var}(\bm{\zeta}_{t})$, $\mH = \mathrm{Var}(\bm{\xi}_{t})$ are diagonal $\forall t = 1,..., T$. Without loss of generality, we set $\mQ = \mI$ for all fractional factor models in state space form to distinguish between the factor loadings $\mLambda$ and the variance of the factor innovations $\mQ$. The matrices $\mZ$, $\mT$, $\mR$, and the states $\boldsymbol{\alpha}_t$ differ for the three fractional factor models and are derived separately in the following.
\paragraph{ARFI factors} By approximating the fractional difference operator of our first model that is given in section \ref{Ch:2}, equation \eqref{eq:2} becomes
\begin{align*}
\zeta_{j, t}=B_j(L)(1-L)^{d_j}f_{j, t} \overset{a}{=} B_j(L)a(L, d_j)_+f_{j, t}, \qquad j=1,...,r, \quad t=1,...,T.
\end{align*}
Let $\mA(L, \vd) = \mI - \sum_{j=1}^v \mA_j(\vd)L^j$, $\mA_j(\vd) = \diag(a_{j}(d_1),...,a_{j}(d_r))$, $j=1,...,r$. Then $\mB(L)\mA(L, \vd)=\sum_{k=0}^{p+v}\sum_{l=0}^{k}\mB_{l}\mA_{k-l}(\vd)L^k$ where $\mA_{0}(\vd)=\mB_{0}=-\mI$, $\mA_l(\vd)=\boldsymbol{0} \ \forall l > v$, and $\mB_l = \boldsymbol{0} \ \forall l > p$. 

\noindent
\cite{JunKoo2015} suggest to eliminate autocorrelation in the idiosyncratic errors $\vu_t$ via the observations equation instead of accounting for them via the state equation. We follow their suggestion and manipulate the observable variables
\begin{align}\label{y:tilde}
\tilde{y}_{i,t} = y_{i,t} - \sum_{j=1}^{p_i}\rho_{i,k}y_{i,t-j}, &&\forall i=1,...,N.
\end{align}
For the state space representation we collect the adjusted observable variables in $\tilde{\vy_t}=(\tilde{y}_{1,t},..., \tilde{y}_{N,t})'$ and define $\bm{\Psi}_j = \mathrm{diag}(\rho_{1,j}, ..., \rho_{N, j})$ such that $\tilde{\vy}_t=\vy_t-\sum_{j=1}^{\mathrm{max}(p_i)}\bm{\Psi}_j\vy_{t-j}.$ 

\noindent
A state space representation of \eqref{eq:1}, \eqref{eq:2}, and \eqref{eq:3} follows directly by defining the system matrices $\mT$, $\mZ$, $\mR$, together with the state vector $\valpha_t$ in \eqref{eq:5} as follows. $\mT$ depends on $\vd$ and $\mB(L)$, whereas $\mZ$ depends on $\mLambda$ and ${\rho_i}(L)$, $i=1,...,N$,
	\begin{align*}
	\mT &= \bmat
	\mB_1+\mA_1(\vd)& \cdots &  -\sum_{l=0}^{u-1}\mB_l\mA_{u-1-l}(\vd)&  -\sum_{l=0}^{u}\mB_l\mA_{u-l}(\vd) \\
	\mI & \cdots & \bm{0}& \bm{0}\\
	\vdots & \ddots & \vdots & \vdots \\
	\bm{0}& \cdots &  \mI & \bm{0} \\
	\emat, &&  \mZ = \bmat
	\boldsymbol{\Lambda}' \\-(\boldsymbol{\Psi}_1\boldsymbol{\Lambda})'\\ \vdots \\ -(\boldsymbol{\Psi}_{u-1}\boldsymbol{\Lambda})'
\emat',
	\end{align*}
$\boldsymbol{\alpha}_t=(\vf_{t}',...,\vf_{t-u+1}')'$ holds the states, $\mR=[\mI, \bm{0}]'$ is a selection matrix and $u$ is defined as $\mathrm{max} (p+v, \mathrm{max} (p_i)+1)$. To distinguish between the $r$ factors, we restrict the first $r$ rows of $\bm{\Lambda}$ to form a lower triangular matrix. 
 \paragraph{FI and AR factors}
Starting with the ARMA approximations of the purely fractional factors in \eqref{DOFC:f1}, an approximate representation of the latent fractionally integated factors is given by
$\vf_t^{(1)}\overset{a}{=}\mM(L, \vd)\mA(L, \vd)^{-1}_+\boldsymbol{\zeta}_t^{(1)}$, where the matrix AR and MA polynomials are $\mM(L, \vd)=\mI + \mM_1(\vd)L + ... + \mM_w(\vd)L^w$, $\mM_j(\vd)=\diag (m_j(d_1),...,m_j(d_{r_1}))$, $\mA(L, \vd)=\mI - \mA_1(\vd)L - ... - \mA_v(\vd)L^v$, $\mA_j(\vd)=\diag (a_j(d_1),...,a_j(d_{r_1}))$, and $\mM_j(\vd) = \bm{0}$ $\forall j > w$, $\mA_j(\vd)=\bm{0}$ $\forall j > v$. 

\noindent 
Regarding $\vu_t$, we again eliminate autocorrelation from the idiosyncratic errors by manipulating $\vy_t$ as in \eqref{y:tilde}, i.e.\ $\tilde{\vy}_t = \vy_t - \sum_{j=1}^{\mathrm{max}(p_i)} \mPsi_j \vy_{t-j}= \mPsi(L)\vy_t$. 
For the latent fractionally integrated factors, this implies $\mPsi(L)\mLambda^{(1)}\vf_t^{(1)}\overset{a}{=}\mPsi(L)\mLambda^{(1)}\mM(L, \vd)\mA(L, \vd)^{-1}_+\boldsymbol{\zeta}_t^{(1)}$. 

\noindent
The state space form \eqref{eq:5} of the model is then obtained by imposing a block diagonal structure on $\mT=\mathrm{diag}(\mT^{(1)}, \mT^{(2)})$, where the first block $\mT^{(1)}$ solely depends on $\vd$, whereas the second block $\mT^{(2)}$ depends on $\mB(L)$
	\begin{align*}
		\mT^{(1)} = \bmat
			\mA_{1}(\vd) & \cdots & \mA_{u_1-1}(\vd) & \mA_{u_1}(\vd) \\
			\mI & \cdots & \bm{0} & \bm{0} \\
			\vdots &  \ddots & \vdots & \vdots \\
			\bm{0} & \cdots & \mI & \bm{0}
		\emat, 
		&&\mT^{(2)} = \bmat
			\mB_{1}^{(2)} & \cdots & \mB_{u_2-1}^{(2)} &\mB_{u_{2}}^{(2)} \\
			\mI &  \cdots & \bm{0} &\bm{0} \\
			\vdots &  \ddots & \vdots & \vdots \\
			\bm{0} &  \cdots & \mI& \bm{0}
		\emat,
	\end{align*}
where $u_1=\mathrm{max} (v, w+\mathrm{max} (p_i)+1)$, $u_2=\mathrm{max}(p, \mathrm{max} (p_i)+1)$, and $\mathrm{max}(p_i)$ is the maximum lag order of the idiosyncratic errors $\vu_t$ in \eqref{DOFC:u}. $\mT^{(1)}$ accounts for the dynamic properties of the fractionally integrated factors, whereas $\mT^{(2)}$ models the stationary variation of the $\vf_t^{(2)}$. 

\noindent
The two blocks for $\mZ=\bmat
 \mZ^{(1)} & \mZ^{(2)}
 \emat$ depend on $\bm{\Lambda}^{(1)}$, $\bm{\Lambda}^{(2)}$, $\vd$, and $\boldsymbol{\rho}(L)$, and are given by
 \begin{align*}
 \mZ^{(1)} &= 
 \bmat
 \bm{\Lambda}^{(1)} & \sum_{k=0}^{1}-\bm{\Psi}_k\bm{\Lambda}^{(1)}\mM_{1-k}(\vd) & \cdots & \sum_{k=0}^{u_1-1}-\bm{\Psi}_k\bm{\Lambda}^{(1)}\mM_{u_1-1-k}(\vd)
 \emat, \\
 \mZ^{(2)} &= \bmat
 \bm{\Lambda}^{(2)} & -\bm{\Psi}_1\bm{\Lambda}^{(2)} & \cdots & -\bm{\Psi}_{u_2-1}\bm{\Lambda}^{(2)}
 \emat,
 \end{align*}
 whereas the state vector is given by
\begin{align*}
	&\bm{\alpha}_{t} =
	\begin{pmatrix}
		\bm{\alpha}^{(1)}_{t} \\
		\bm{\alpha}^{(2)}_{t}
	\end{pmatrix}, 
	&&\bm{\alpha}_{t}^{(1)} = \left( \mI -... - \mA_{v}(\vd)L^{v}\right)^{-1}_{+} \begin{pmatrix}
		\bm{\zeta}_{t}^{(1)} \\
		\vdots \\
		\bm{\zeta}_{t-u_1+1}^{(1)} \\
	\end{pmatrix}, 
	&&&\bm{\alpha}_{t}^{(2)} = \begin{pmatrix}
		\vf^{(2)}_{t} \\
		\vdots \\
		\vf^{(2)}_{t-u_2+1}
	\end{pmatrix}.
\end{align*}
Note that $\bm{\Psi}_j = \bm{0} \ \forall j > \mathrm{max}(p_i)$ and $\mB_j^{(2)}=\bm{0}\ \forall j>p$. Finally, the selection matrices are given by $\mR = \mathrm{diag}(\mR^{(1)}, \mR^{(2)})$, $\mR^{(1)} = [\mI, \bm{0}]'$, $\mR^{(2)} = [\mI, \bm{0}]'$, whereas the disturbances in the state equation are $\bm{\zeta}_{t+1} = (\bm{\zeta}_{t+1}^{(1)'}, \bm{\zeta}_{t+1}^{(2)'})'$, $\bm{\zeta}_{t} \sim \mathrm{NID}(\bm{0}, \mQ)$ for all $t = 1,...,T$. Note that for the fractionally integrated factors the observations equation yields $\mZ^{(1)}\valpha_t^{(1)} = \mPsi(L)\mLambda^{(1)} \mM(L, \vd) \valpha_t^{(1)} =\mPsi(L)\mLambda^{(1)} \mM(L, \vd)\mA(L, \vd)_+ \vzeta_{t}^{(1)} \overset{a}{=} \mPsi(L)\mLambda^{(1)}\vf_t^{(1)}$, whereas for the stationary AR factors it gives $\mZ^{(2)}\valpha_t^{(2)} = \mPsi(L)\mLambda^{(2)}\vf_t^{(2)}$.
\\
The $r_1$ independent fractional factors are identified by
imposing a block triangular structure on $\bm{\Lambda}^{(1)}$ while sorting the observations $\vy_t$ with respect to their order of fractional integration in ascending order. As a consequence, the first block of variables in $\vy_t$ is driven by the least persistent factor $f_{1t}$, the second block of variables depends on $f_{1t}$ and $f_{2t}$ whereas the $r_1$-th block with the highest order of fractional integration is allowed to be influenced by all fractional factors. In addition, the first $r_2$ rows of $\bm{\Lambda}^{(2)}$ form a lower triangular matrix to identify the $I(0)$ factors $\vf_t^{(2)}$.
 
 \noindent
 \paragraph{AR factors}
 Since the factors of our third model \eqref{sfd:1} are stationary autoregressive processes, a state space representation as in \eqref{eq:5} follows immediately by defining $\tilde{\vy}_t = (\Delta_+^{d_1^*}y_{1, t},...,\Delta_+^{d_N^*}y_{N, t})'$. The factors enter the state vector directly, whereas their dynamic coefficients in \eqref{sfd:2} are contained in $\mT$. Furthermore, the factor loadings are modelled via $\mZ$, and $\mR$ is again a selection matrix
 	\begin{align*}
 	\bm{\alpha}_{t} = \begin{pmatrix}
 	\vf_{t} \\
 	\vf_{t-1} \\
 	\vdots \\
 	\vf_{t-p}
 	\end{pmatrix}, \qquad
 	\mT = \bmat
 	\mB_1 & \cdots & \mB_{p} & \bm{0} \\
 	\mI&  \cdots & \bm{0} & \bm{0} \\
 	\vdots & \ddots & \vdots & \vdots \\
 	\bm{0} & \cdots & \bm{I} & \bm{0}
 	\emat, \qquad 
 	\mZ = \bmat
 	\bm{\Lambda}'\\ \bm{0} \\ \vdots \\ \bm{0}
 	\emat', \qquad
 	\bm{R} = \bmat
 	\bm{I} \\
 	\bm{0}
 	\emat.
 	\end{align*}
For identification of the factors, we restrict the first $r$ rows of $\bm{\Lambda}$ to be lower triangular.

\noindent
\subsubsection{Parameter Estimation}\label{Ch:3.2.3}
We collect the unknown parameters in $\vd$, ${\mLambda}$, $\mB_1, ..., \mB_p$, $\rho_{1,1}, ..., \rho_{N,p_N}$, and $\mH$, that enter the system matrices of the state space model $\mT$, $\mZ$, and $\mH$, in a parameter vector $\boldsymbol{\theta}$. To estimate ${\vtheta}$ we adopt the approach of \cite{HarWei2018a}, who derive an analytical solution to the optimization problem of the expected complete Gaussian likelihood function of the state space model, together with a computationally fast combination of the EM algorithm and gradient-based optimization. 

\noindent
In the expectation step of the EM algorithm, we estimate the smoothed states and disturbances, together with the corresponding covariance matrices for a given set of parameters $\hat{\vtheta}_j$ via the Kalman filter and smoother. 
The M-step then maximizes the likelihood given the Kalman filter and smoother estimates to obtain $\hat{\boldsymbol{\theta}}_{j+1}$. 

\noindent
After either a convergence criterion is satisfied, or a predefined number of iterations $m$ is reached, the resulting parameter estimates from the EM algorithm $\hat{\bm{\theta}}$ are used as starting values for the maximum likelihood estimation via the BFGS algorithm, which uses the analytical solution for the score vector of \cite{HarWei2018a}, since the EM algorithm was found to be slow around the optimum. 

\noindent
In case of the stationary factor model in fractional differences, the matrices $\mT$ and $\mZ$ are functions of two disjoint parameter spaces and, therefore, a simplification of the EM algorithm is obtained directly by solving the score vector for $\mathrm{vec}(\mT)$ and $\mathrm{vec}(\mZ)$  \citep[see][appendix A.2]{JunKoo2015}. 

\noindent
Forecasts are obtained by shifting the system one period ahead and plugging in the smoothed factor estimates from the Kalman filter.
\section{Macroeconomic Forecasting}\label{Ch:5}
\subsection{Forecast Design}
Having discussed the estimation of $f(\vchi_t)$ together with the unknown parameters for the three fractionally integrated factor models in sections \ref{Ch:2}--\ref{Ch:4}, we investigate their predictive accuracy when neither the DGP, nor the starting values, nor the number of factors, are known to the researcher. 
For this purpose we study the forecast performance of our three models in a pseudo out-of-sample forecast experiment with an underlying data set for the United States of America that consists of 112 macroeconomic variables and spans from January 1960 to December 2016 \citep[see][]{McCNg2015}. 

\noindent
To compare the forecast performance of the different factor models, we report the resulting mean squared prediction errors (MSPE) for a selected subset of economic variables that represent different segments of the economy. 

\noindent
All forecast models allow for seven common factors in the data, which is suggested by the $PC_{(p3)}$ criterion of \cite{BaiNg2002} after deterministic terms have been eliminated from the fractionally differenced data set. Lag lengths of the different AR polynomials for the common factors and idiosyncratic components are chosen via the Bayesian Information Criterion (BIC). For the first forecasting period we obtain starting values for the Kalman filter from the principal components estimator, as described in section \ref{sec2.1}. In all subsequent periods the optimized parameters from the preceding step are used as starting values. Finally, the number of iterations of the EM algorithm is set to ten.
To distinguish between the first stage and the second stage estimator, we denote the principal components forecasts as \textbf{PC} and the Kalman filter forecasts as \textbf{KF}. Abbreviations for the three fractional factor models are: Dynamic fractional factor model (\textbf{DFFM}) in section \ref{Ch:2}, dynamic orthogonal fractional components (\textbf{DOFC}) in section \ref{Ch:3}, and dynamic factor model in fractional differences (\textbf{DFFD}) in section \ref{Ch:4}.

\noindent
Forecasts are conducted for horizons $h=1,...,12$ in a recursive window forecast experiment, where the first forecast period is January 2000, whereas the last is December 2016, leading to 204 forecasts for 112 variables and 12 horizons.

\noindent
The \textbf{DOFC} model introduced in section \ref{Ch:3} includes $r_1=3$ fractionally integrated factors, since a higher number was not found to increase the forecast precision substantially. As a consequence, the number of remaining $I(0)$ factors is set to $r_2=4$. The latter restriction is confirmed by the $PC_{(p3)}$ criterion of \cite{BaiNg2002}, which suggests four factors after the fractionally integrated factors have been projected out. 

\noindent
A stationary data set for the \textbf{DFFD} model in section \ref{Ch:4} is obtained by estimating the integration order of each $y_{i,t}$ via the exact local Whittle estimator of \cite{ShiPhi2005} with a tuning parameter of $0.5$ and taking fractional differences.

\noindent
In addition, we include four benchmark models to evaluate the forecast performance of the fractional factor models relative to widely used alternatives. The first benchmark is an autoregressive model (\textbf{AR}) where the AR lag order is chosen via the Akaike Information Criterion for each $y_{i,t}$.
The second benchmark is a standard approximate dynamic factor model \citep[cf.\ e.g.][]{StoWat2002} that is estimated via principal components (\textbf{PC}) based on a pre-differenced data set, i.e. $\Delta^{k_i}y_{i,t+h}=\bm{\Lambda}_i\vf_{t+h}+\xi_{i,t+h}$, $\bm{\phi}(L)\bm{f}_{t+h}= \bm{\zeta}_{t+h}$,  where $\xi_{i,t}$, $\zeta_{j,t}$ are mutually independent and white noise $\forall t=1,...,T$ and $k_i$ is an integer that is taken from \cite{McCNg2015}. Our third model adds lagged dependent variables to the approximate dynamic factor model. It is given by $\bm{\phi}(L)\bm{f}_{t+h}= \bm{\zeta}_{t+h}$, $c_i(L)\Delta^{k_i}y_{i,t+h}=\bm{\Lambda}_i\vf_{t+h}+\xi_{i,t+h}$ where $\xi_{i,t}$, $\zeta_{j,t}$ are again mutually independent and white noise $\forall t=1,...,T$. We denote it as \textbf{PCAR}. Finally, the last benchmark is the so-called factor-augmented error-correction model (\textbf{FECM}), which separates the observable variables into two disjoint samples $\vy = (\vy^{(1)'}, \vy^{(2)'})'$ and shrinks the latter sample via principal components to $\hat{\vf}$. A vector error-correction model is then estimated for $(\vy^{(1)'}, \hat{\vf}')'$. Details on the forecast properties are found in \cite{BanMarMa2014}. Since we only obtain predictions for $\vy^{(1)}$, the $\textbf{FECM}$ results are only reported in tables \ref{ta:1} and \ref{ta:2}.

\subsection{Forecast Results}
Table \ref{ta:3} shows for a given forecast horizon $h$ how often each specification leads to the smallest MSPE for all $112$ variables. Hence, it illustrates how frequently fractional factor models are able to outperform widely used forecast methods like autoregressive models and principal components of integer differences. 
To draw inference on the extent of forecast improvement from the fractional factor models, the tables \ref{ta:1} and \ref{ta:2} report the relative MSPE for the twelve depicted variables and for $h = 1, 2, 3, 6, 9, \text{ and }12$. Consequently, they also show how large the forecast accuracy fluctuates for each specification and highlight the robustness of the forecast results when a model is not chosen to be the best one. 

\linespread{0.5}{
\begin{table}[h!]
	\centering
	\begin{tabular}{l|rrrrrrrrr}
		\hline
		&\multicolumn{3}{c|}{Benchmarks} & \multicolumn{2}{c|}{DFFM} & \multicolumn{2}{c|}{DOFC} & \multicolumn{2}{c}{DFFD} \\
		Horizon& \multicolumn{1}{r}{AR} & {PC} & \multicolumn{1}{r|}{PCAR} & {PC} & \multicolumn{1}{r|}{KF} & {PC} & \multicolumn{1}{r|}{KF} & {PC} & \multicolumn{1}{r}{KF}  \\ 
		\hline
	1 &  14 &  10 &  16 &   7 &   0 &   3 &  26 &  20 &  16 \\ 
	2 &  14 &  12 &  11 &   8 &   0 &   1 &  26 &  24 &  16 \\ 
	3 &  14 &  10 &   6 &   6 &   1 &   5 &  28 &  23 &  19 \\ 
	4 &  17 &   9 &   7 &  11 &   2 &   4 &  25 &  20 &  17 \\ 
	5 &  15 &  11 &   7 &  10 &   1 &   5 &  27 &  22 &  14 \\ 
	6 &  17 &   8 &   5 &  10 &   6 &   6 &  23 &  19 &  18 \\ 
	7 &  15 &   9 &   3 &  10 &   4 &   7 &  26 &  18 &  20 \\ 
	8 &  14 &   8 &   3 &  10 &  12 &   7 &  20 &  19 &  19 \\ 
	9 &  15 &   8 &   3 &  10 &  10 &   7 &  22 &  17 &  20 \\ 
	10 &  15 &   8 &   3 &  11 &  14 &   7 &  16 &  15 &  23 \\ 
	11 &  15 &   8 &   3 &  13 &  15 &   7 &  13 &  15 &  23 \\ 
	12 &  14 &   7 &   3 &  10 &  17 &   9 &  13 &  16 &  23 \\ 
		\hline
	\end{tabular}
\caption{Frequency of smallest MSPE: The table shows how often, for a given forecast horizon, a specification came with the smallest mean squared prediction error of all models. }
\label{ta:3}
\end{table}}%
\noindent
We find that fractional factor models tend to outperform classical autoregressive models, pre-differenced principal components models and mixtures of these two model classes. Over all 1344 conducted forecasts, the benchmarks only exhibit a smaller MSPE than the fractional factor models in 357 cases (26.6\%), as table \ref{ta:3} shows. Hence, for the remaining 987 forecasts (73.4 \%) the smallest MSPE is achieved by one of the six fractional factor models. Within the benchmarks, one often finds that principal components come with the smallest MSPE. Nonetheless, they are often beaten by one of the fractional factor models. Among those, the dynamic orthogonal fractional components model in state space form produces the best predictions for forecast horizons up to 9 months most frequently. 

\noindent
In addition to the good performance of the DOFC-KF specification, the DFFD models complement the predictive power of fractional factor models. Whenever the DOFC-KF model does not provide the best forecasts, the fractionally differenced models are likely to exhibit the smallest MSPE. Furthermore, principal components are found to perform relatively well at least for smaller forecast horizons when the data is in fractional differences, whereas they are typically beaten by the state space formulation in the DOFC framework. This might be a result of the additional structure that is imposed on the DOFC-KF model via the block-triangular identification of the fractional factors relative to the DOFC-PC case, whereas only little additional structure is imposed on the DFFD-KF specification relative to principal components. For larger forecast horizons, the forecast performance of the DFFD-KF model improves, leading to the highest amount of best predictions for $h = 10, 11, 12$. 

\linespread{0.5}{
\begin{table}[p]
	\centering
	\resizebox{15cm}{!}{%
	\begin{tabular}{lllllrrrrrr}
		\hline
		&\multicolumn{4}{c|}{Benchmarks} & \multicolumn{2}{c|}{DFFM} & \multicolumn{2}{c|}{DOFC} & \multicolumn{2}{c}{DFFD} \\
		& \multicolumn{1}{r}{AR} & {PC} & PCAR & \multicolumn{1}{r|}{FECM} & {PC} & \multicolumn{1}{r|}{KF} & {PC} & \multicolumn{1}{r|}{KF} & {PC} & \multicolumn{1}{r}{KF}  \\ 
		\hline
				\multicolumn{11}{c}{Horizon h = 1} \\
		INDPRO & 1.00 & 0.92 & 0.92 & 1.22 & 1.60 & 3.26 & 1.96 & 1.06 & \underline{0.90} & 0.99 \\ 
		UNRATE & 1.00 & 0.96 & 0.95 & 0.88 & 1.06 & 2.67 & 1.32 & 0.91 & 0.89 & \underline{0.88} \\ 
		AWOTMAN & 1.00 & \underline{0.85} & 0.90 & 0.92 & 1.25 & 2.02 & 1.40 & 0.91 & 0.93 & 0.96 \\ 
		HOUST & 1.00 & 0.71 & \underline{0.70} & 0.99 & 0.98 & 1.36 & 0.94 & 0.87 & 1.09 & 1.07 \\ 
		AMBSL & 1.00 & 0.86 & 1.17 & \underline{0.69} & 0.77 & 2.62 & 0.77 & 0.81 & 2.78 & 3.15 \\ 
		TOTRESNS & 1.00 & 0.71 & 1.45 & 0.69 & \underline{0.66} & 2.28 & 0.70 & 0.69 & 4.84 & 5.14 \\ 
		S.P.500 & 1.00 & 1.08 & 1.08 & 1.01 & 1.04 & 2.67 & 1.24 & 1.02 & 1.06 & \underline{0.99} \\ 
		FEDFUNDS & \underline{1.00} & 2.51 & 2.36 & 2.64 & 1.15 & 4.43 & 1.17 & 1.21 & 3.53 & 1.25 \\ 
		EXUSUKx & \underline{1.00} & 1.18 & 1.08 & 1.08 & 1.11 & 2.75 & 1.17 & 1.06 & 1.12 & 1.10 \\ 
		CPIAUCSL & 1.00 & 0.64 & 0.96 & 0.45 & 1.77 & 2.01 & 1.58 & \underline{0.41} & 0.49 & 0.49 \\ 
		PCEPI & 1.00 & 0.69 & 0.98 & 0.50 & 3.32 & 1.98 & 2.70 & \underline{0.41} & 0.47 & 0.47 \\ 
		CES0600000008 & 1.00 & 0.88 & 1.14 & 0.48 & 2.50 & 1.00 & 3.06 & \underline{0.30} & 0.48 & 0.41 \\ 
				\multicolumn{11}{c}{Horizon h = 2} \\
		INDPRO & 1.00 & 0.90 & 0.90 & 1.35 & 2.10 & 2.38 & 2.63 & 1.15 & \underline{0.81} & 0.99 \\ 
		UNRATE & 1.00 & 1.03 & 0.98 & \underline{0.79} & 1.06 & 2.10 & 1.62 & 0.96 & 0.83 & 0.86 \\ 
		AWOTMAN & 1.00 & 0.81 & 0.92 & \underline{0.80} & 1.22 & 1.69 & 1.50 & 0.86 & 0.96 & 1.02 \\ 
		HOUST & 1.00 & \underline{0.71} & 0.71 & 1.01 & 0.99 & 1.18 & 0.95 & 0.99 & 1.02 & 1.00 \\ 
		AMBSL & 1.00 & 0.79 & 1.27 & 0.86 & 0.76 & 1.25 & 0.80 & \underline{0.72} & 1.52 & 1.65 \\ 
		TOTRESNS & 1.00 & 0.69 & 1.56 & 0.75 & 0.67 & 1.17 & 0.75 & \underline{0.65} & 2.45 & 2.51 \\ 
		S.P.500 & 1.00 & 1.14 & 1.16 & 1.18 & 1.05 & 1.55 & 1.30 & 1.02 & 1.11 & \underline{0.98} \\ 
		FEDFUNDS & 1.00 & 1.66 & 1.79 & 2.66 & \underline{0.91} & 2.08 & 0.92 & 0.95 & 2.40 & 1.07 \\ 
		EXUSUKx & \underline{1.00} & 1.20 & 1.13 & 1.15 & 1.18 & 1.71 & 1.22 & 1.05 & 1.11 & 1.08 \\ 
		CPIAUCSL & 1.00 & 0.61 & 0.98 & 0.54 & 1.76 & 0.50 & 1.62 & \underline{0.42} & 0.60 & 0.55 \\ 
		PCEPI & 1.00 & 0.62 & 0.99 & 0.56 & 3.24 & 0.45 & 2.66 & \underline{0.39} & 0.52 & 0.48 \\ 
		CES0600000008 & 1.00 & 0.99 & 1.16 & 0.37 & 2.16 & 0.64 & 3.26 & \underline{0.21} & 0.34 & 0.28 \\ 
				\multicolumn{11}{c}{Horizon h = 3} \\
	INDPRO & 1.00 & 1.04 & 1.04 & 1.50 & 2.41 & 2.47 & 2.93 & 1.28 & \underline{0.81} & 1.00 \\ 
	UNRATE & 1.00 & 1.15 & 1.08 & \underline{0.82} & 1.11 & 2.21 & 1.76 & 1.04 & 0.82 & 0.88 \\ 
	AWOTMAN & 1.00 & 0.79 & 0.95 & \underline{0.76} & 1.10 & 1.56 & 1.47 & 0.84 & 1.03 & 1.02 \\ 
	HOUST & 1.00 & \underline{0.70} & 0.70 & 1.01 & 0.88 & 1.11 & 0.85 & 0.95 & 1.09 & 0.99 \\ 
	AMBSL & 1.00 & 0.77 & 1.41 & 0.97 & 0.74 & 0.95 & 0.81 & \underline{0.66} & 1.16 & 1.22 \\ 
	TOTRESNS & 1.00 & 0.73 & 1.62 & 0.81 & 0.69 & 0.94 & 0.78 & \underline{0.64} & 1.81 & 1.81 \\ 
	S.P.500 & 1.00 & 1.25 & 1.27 & 1.30 & 1.08 & 1.45 & 1.35 & 1.02 & 1.18 & \underline{1.00} \\ 
	FEDFUNDS & 1.00 & 1.31 & 1.41 & 2.66 & 0.85 & 1.54 & \underline{0.82} & 0.91 & 1.92 & 1.03 \\ 
	EXUSUKx & \underline{1.00} & 1.26 & 1.18 & 1.21 & 1.22 & 1.55 & 1.29 & 1.07 & 1.12 & 1.08 \\ 
	CPIAUCSL & 1.00 & 0.57 & 1.01 & 0.54 & 1.71 & 0.47 & 1.63 & \underline{0.38} & 0.63 & 0.56 \\ 
	PCEPI & 1.00 & 0.60 & 1.01 & 0.58 & 3.12 & 0.43 & 2.72 & \underline{0.36} & 0.53 & 0.47 \\ 
	CES0600000008 & 1.00 & 1.13 & 1.18 & 0.40 & 1.79 & 0.56 & 2.95 & \underline{0.18} & 0.27 & 0.22 \\ 
	\end{tabular}
}%
\noindent
\caption{Selected relative mean squared prediction errors for h=1, 2, and 3. Variable codes are INDPRO: industrial production index; UNRATE: unemployment rate; AWOTMAN: average weekly overtime hours in the manufacturing business; HOUST: housing starts; AMBSL: St. Louis adjusted monetary base; TOTRESNS: total reserves of depository institutions; S.P.500: S\&P500 index; FEDFUNDS: effective federal funds rate; EXUSUKx: US / UK foreign exchange rate; CPIAUCSL: consumer price index; PCEPI: personal consumption index; CES0600000008: average hourly earnings}
\label{ta:1}
\end{table}
\begin{table}[p]
	\centering
	\resizebox{15cm}{!}{%
		\begin{tabular}{lllllrrrrrr}
			\hline
			&\multicolumn{4}{c|}{Benchmarks} & \multicolumn{2}{c|}{DFFM} & \multicolumn{2}{c|}{DOFC} & \multicolumn{2}{c}{DFFD} \\
			& \multicolumn{1}{r}{AR} & {PC} & {PCAR} &\multicolumn{1}{r|}{FECM} & {PC} & \multicolumn{1}{r|}{KF} & {PC} & \multicolumn{1}{r|}{KF} & {PC} & \multicolumn{1}{r}{KF}  \\ 
			\hline
					\multicolumn{11}{c}{Horizon h = 6} \\
		INDPRO & 1.00 & 1.27 & 1.27 & 1.59 & 2.10 & 1.82 & 2.55 & 1.29 & \underline{0.95} & 1.08 \\ 
		UNRATE & \underline{1.00} & 1.58 & 1.46 & 1.11 & 1.25 & 2.04 & 1.86 & 1.25 & 1.05 & 1.06 \\ 
		AWOTMAN & 1.00 & 0.92 & 1.12 & \underline{0.76} & 1.03 & 1.33 & 1.40 & 0.87 & 1.17 & 1.13 \\ 
		HOUST & 1.00 & \underline{0.63} & 0.63 & 1.02 & 0.77 & 0.88 & 0.78 & 0.82 & 0.98 & 0.87 \\ 
		AMBSL & 1.00 & 0.89 & 1.87 & 0.87 & 0.62 & 0.68 & 0.76 & \underline{0.55} & 0.89 & 0.90 \\ 
		TOTRESNS & 1.00 & 0.95 & 1.62 & 0.75 & 0.63 & 0.68 & 0.76 & \underline{0.57} & 1.18 & 1.16 \\ 
		S.P.500 & \underline{1.00} & 1.50 & 1.53 & 1.40 & 1.09 & 1.21 & 1.47 & 1.00 & 1.19 & 1.01 \\ 
		FEDFUNDS & 1.00 & 1.28 & 1.33 & 2.70 & 0.90 & 1.21 & \underline{0.76} & 0.91 & 1.36 & 1.06 \\ 
		EXUSUKx & \underline{1.00} & 1.31 & 1.30 & 1.46 & 1.23 & 1.18 & 1.35 & 1.03 & 1.06 & 1.00 \\ 
		CPIAUCSL & 1.00 & 0.57 & 1.09 & 0.46 & 1.49 & \underline{0.26} & 1.65 & 0.31 & 0.63 & 0.54 \\ 
		PCEPI & 1.00 & 0.59 & 1.06 & 0.54 & 2.69 & \underline{0.25} & 2.85 & 0.34 & 0.52 & 0.44 \\ 
		CES0600000008 & 1.00 & 1.82 & 1.27 & 0.54 & 0.86 & 0.50 & 2.38 & 0.17 & 0.21 & \underline{0.16}  \\ 
					\multicolumn{11}{c}{Horizon h = 9} \\
		INDPRO & 1.00 & 1.43 & 1.43 & 1.85 & 1.83 & 1.48 & 2.30 & 1.25 & \underline{1.00} & 1.12 \\ 
		UNRATE & \underline{1.00} & 1.88 & 1.77 & 1.48 & 1.23 & 1.76 & 1.65 & 1.26 & 1.17 & 1.12 \\ 
		AWOTMAN & 1.00 & 1.11 & 1.33 & \underline{0.74} & 1.01 & 1.20 & 1.34 & 0.91 & 1.21 & 1.14 \\ 
		HOUST & 1.00 & \underline{0.60} & 0.61 & 1.02 & 0.72 & 0.78 & 0.73 & 0.76 & 0.94 & 0.85 \\ 
		AMBSL & 1.00 & 1.31 & 3.30 & 0.91 & 0.59 & 0.59 & 0.79 & \underline{0.51} & 0.81 & 0.80 \\ 
		TOTRESNS & 1.00 & 1.54 & 1.74 & 0.73 & 0.61 & 0.60 & 0.76 & \underline{0.53} & 0.97 & 0.94 \\ 
		S.P.500 & 1.00 & 1.76 & 1.79 & 1.50 & 1.09 & 1.13 & 1.54 & \underline{0.99} & 1.19 & 1.01 \\ 
		FEDFUNDS & 1.00 & 1.41 & 1.50 & 2.55 & 0.93 & 1.16 & \underline{0.79} & 0.96 & 1.18 & 1.07 \\ 
		EXUSUKx & 1.00 & 1.40 & 1.41 & 1.78 & 1.29 & 1.08 & 1.42 & 1.01 & 1.05 & \underline{0.98} \\ 
		CPIAUCSL & 1.00 & 0.64 & 1.13 & 0.43 & 1.32 & \underline{0.21} & 1.56 & 0.31 & 0.60 & 0.51 \\ 
		PCEPI & 1.00 & 0.65 & 1.10 & 0.53 & 2.40 & \underline{0.21} & 2.70 & 0.37 & 0.50 & 0.42 \\ 
		CES0600000008 & 1.00 & 2.65 & 1.36 & 0.60 & 0.49 & 0.46 & 1.81 & 0.19 & 0.16 & \underline{0.12} \\ 
						\multicolumn{11}{c}{Horizon h = 12} \\
		INDPRO & 1.00 & 1.57 & 1.58 & 2.00 & 1.70 & 1.28 & 2.20 & 1.20 & \underline{0.99} & 1.11 \\ 
		UNRATE & \underline{1.00} & 2.20 & 2.09 & 1.68 & 1.16 & 1.50 & 1.45 & 1.22 & 1.18 & 1.12 \\ 
		AWOTMAN & 1.00 & 1.26 & 1.50 & \underline{0.78} & 0.99 & 1.12 & 1.27 & 0.94 & 1.17 & 1.10 \\ 
		HOUST & 1.00 & \underline{0.61} & 0.61 & 1.11 & 0.72 & 0.74 & 0.71 & 0.74 & 0.94 & 0.86 \\ 
		AMBSL & 1.00 & 1.84 & 6.01 & 0.77 & 0.52 & 0.52 & 0.75 & \underline{0.45} & 0.70 & 0.69 \\ 
		TOTRESNS & 1.00 & 2.36 & 1.93 & 0.63 & 0.55 & 0.53 & 0.70 & \underline{0.48} & 0.80 & 0.77 \\ 
		S.P.500 & 1.00 & 2.00 & 2.03 & 1.57 & 1.07 & 1.09 & 1.59 & \underline{0.98} & 1.20 & 1.00 \\ 
		FEDFUNDS & 1.00 & 1.54 & 1.65 & 2.39 & 0.95 & 1.11 & \underline{0.83} & 1.00 & 1.10 & 1.06 \\ 
		EXUSUKx & 1.00 & 1.53 & 1.58 & 2.01 & 1.36 & 1.05 & 1.48 & 0.99 & 1.05 & \underline{0.98} \\ 
		CPIAUCSL & 1.00 & 0.79 & 1.18 & 0.40 & 1.14 & \underline{0.15} & 1.43 & 0.34 & 0.54 & 0.46 \\ 
		PCEPI & 1.00 & 0.76 & 1.14 & 0.51 & 2.08 & \underline{0.16} & 2.47 & 0.42 & 0.46 & 0.39 \\ 
		CES0600000008 & 1.00 & 4.01 & 1.47 & 0.65 & 0.36 & 0.45 & 1.52 & 0.23 & 0.15 & \underline{0.11}  \\ 
		\end{tabular}
	}%
	\caption{Selected relative mean squared prediction errors for h=6, 9, and 12. Variable codes are INDPRO: industrial production index; UNRATE: unemployment rate; AWOTMAN: average weekly overtime hours in the manufacturing business; HOUST: housing starts; AMBSL: St. Louis adjusted monetary base; TOTRESNS: total reserves of depository institutions; S.P.500: S\&P500 index; FEDFUNDS: effective federal funds rate; EXUSUKx: US / UK foreign exchange rate; CPIAUCSL: consumer price index; PCEPI: personal consumption index; CES0600000008: average hourly earnings}
	\label{ta:2}
\end{table}}%
\noindent
We are able to uncover more details about the forecast performance of fractional factor models by having a closer look at the tables \ref{ta:1} and \ref{ta:2} that visualize the relative MSPEs for selected variables and forecast horizons $h$. We find that gains from the fractional factor models can be large, relative to the four benchmarks. In many cases, fractional factor models can reduce the MSPE relative to the AR benchmark by more than 25\%.  For some target variables, the MSPE is cut by half when fractional factor models are used, and reductions of more than 80\% are possible. 

\noindent
Within the class of fractional factor models, we find the DOFC-KF specification to perform best. For $h=1,2,3$, the most accurate predictions for the consumer price index, personal consumption index and average hourly earnings are obtained from the DOFC-KF specification, which reduces the MSPE relative to the AR benchmark by more than 50\%. In addition, the DOFC-KF specification exhibits the smallest MSPE for the St. Louis adjusted monetary base, total reserves of depository institutions and the S\&P500 frequently. The stable and reliable performance of the DOFC-KF forecasts is illustrated by the fact that their largest relative MSPE is 1.29, whereas the smallest relative MSPE is 0.17. 

\noindent
The good performance of the DOFC-KF specification is complemented by the DFFD model. For the industrial production index, the DFFD-PC specification exhibits the smallest MSPE for any forecast horizon. In addition, the DFFD-KF specification produces accurate predictions for the S\&P500, average hourly earnings and the US / UK foreign exchange rate. Furthermore, its forecast performance is almost as stable as the DOFC-KF prediction quality. 

\noindent
Finally, the DFFM model, which serves as the most general framework as it nests the two remaining fractional factor model formulations, cannot compete with the other fractional factor models, as its predictive power fluctuates largely. Nonetheless, for larger forecast horizons, the DFFM-KF formulation produces accurate forecasts for the consumer price and personal consumption index.

\noindent
Note that the only difference between the benchmark PC model and the DFFD-PC specification is the pre-differencing. As one can see, the two models coincide regarding their relative performance to the AR benchmark. 
 The advantages over the AR model are therefore likely to result from cross-sectional dependencies that are detected by the common factors. In addition, the better performance of the DFFD-PC model can be explained by the sensitivity of standard PC methods to spurious coefficients, as \cite{Fra2017} argue. 

\noindent
Turning to the DOFC-KF specification, which explicitly models fractional cointegration relations instead of eliminating them as in the DFFD model, we note that the forecast quality of the two models is similar for many predictions. Nonetheless, gains from the DOFC-KF specification relative to the DFFD model can be large, especially in situations where the latter produces a relative MSPE $> 1$. Consider e.g.\ the forecasts for the adjusted monetary base (AMBSL) and the total reserves of depository institutions (TOTRESNS) in tables \ref{ta:1} and \ref{ta:2}, where the DOFC-KF and the FECM model perform well, wheres the DFFD-KF model yields large MSPEs. While the former two models take cointegration into account, the DFFD-KF model eliminates long-run components by prior differencing and is likely to produce over-differenced short-run components. Hence, the better performance of the DOFC-KF model over the DFFD-KF model is likely to result from cointegration relations and over-differencing of additive short-run factors. 

\begin{figure}[hp!]
	\includegraphics[scale=0.58, trim={1.6cm 0cm 1cm 1cm}]{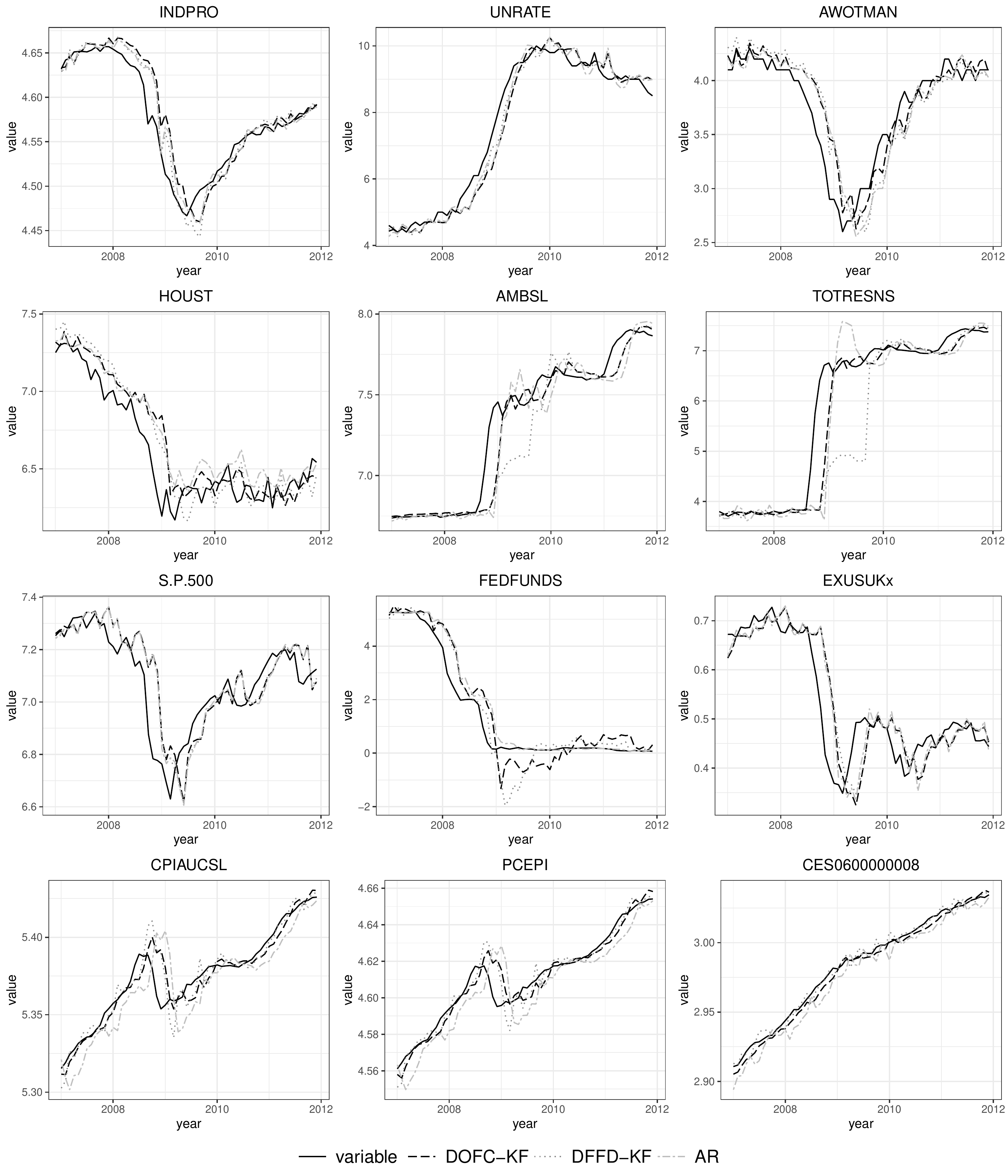}
	\caption[Forecast performance during crysis]{Forecast performance of the DOFC-KF, DFFD-KF and AR model during the world financial crisis 2007 - 2010 for $h=3$. Variable codes are INDPRO: industrial production index; UNRATE: unemployment rate; AWOTMAN: average weekly overtime hours in the manufacturing business; HOUST: housing starts; AMBSL: St. Louis adjusted monetary base; TOTRESNS: total reserves of depository institutions; S.P.500: S\&P500 index; FEDFUNDS: effective federal funds rate; EXUSUKx: US / UK foreign exchange rate; CPIAUCSL: consumer price index; PCEPI: personal consumption index; CES0600000008: average hourly earnings}
	\label{graus:1}
\end{figure}%

\noindent
Finally, we want to draw inference on the performance of the fractional factor models during the world financial crisis. By studying the predictive power of the fractional factor models during this period, we shed light on the behavior of this model class when the economy is hit by a large shock and pushed out of its equilibrium growth path. 
For this purpose, figure \ref{graus:1} sketches the three step ahead predictions for the twelve selected target variables and the two best performing fractional factor models together with the AR benchmark and the realization of the target variable from January 2007 to December 2011. As the graphs show, the forecast performance of the fractional factor models is not systematically affected by the global financial crisis relative to the AR benchmark. Instead, the forecasts converge towards the realizations of the observable variables rapidly after the crisis. The DOFC-KF forecasts seem to be the least affected by the large shock, as they converge faster towards the observable variables. Furthermore, the AR and DFFD-KF predictions for the adjusted monetary base and total reserves of depository institutions seem to be polluted by the crisis until the end of 2009, which substantiates the relative robustness of the DOFC-KF specification. 
\section{Conclusion}\label{Ch:6}
We have derived three different fractional factor models that allow the joint modelling of data of different persistence. A two-stage estimator for the fractional factors and model parameters was derived. In a macroeconomic forecast experiment, it was shown that incorporating fractional integration into the class of factor models improves forecast performance substantially. \\
Future research could examine whether a combination of the DOFC model in state space form and a factor model in fractional differences can improve the predictive power of fractional factor models. Furthermore, one could combine principal components and the Kalman filter analogous to \cite{BraeKoo2014} by reducing the dimension of a subset of observable variables via principal components in order to speed up the estimation of the parameters. Additionally, fractional factor models could be used to explore common trends and cycles in macroeconomic variables and to identify cointegrated blocks. Finally, future research could address the predictive power of fractional factor models for other data sets and economies. If gains are of similar size as for the US, we are confident that fractional factor models have the potential to become a widely used tool for predicting macroeconomic dynamics.

\newpage
\appendix
\section{Consistency of Principal Components for Fractionally Integrated Data} \label{sec:A1}
To proof consistency of principal components in a fractionally integrated setup, we define a minimal fractionally integrated factor model that is given by 
\begin{align}\label{eq:1a}
\vy_t &= \boldsymbol{\Lambda}\vf_t+\ve_{t}, \hspace{0.5cm} &&\forall t=1,...,T,\\
(1-L)^{d}f_{j, t}&=z_{j, t}, && \forall j=1,...,r,
\end{align}
where $\vy_t=(y_{1, t},...,y_{N, t})'$ is a $N$-dimensional vector holding the observable data at point $t$, $\vf_t=(f_{1,t},...,f_{r,t})'\sim I(d)$ contains the unobserved common factors and is $r \times 1$, and $\bm{\Lambda}$ is $N \times r$ and holds the factor loadings. The $N \times 1$ vector $\ve_t=(e_{1,t},...,e_{N,t})'$ and the $r \times 1$ vector $\vz_t=(z_{1,t},...,z_{r,t})'$ are $I(0)$ stochastic processes, $z_{j, t} = \sum_{k=0}^{t-1}c_{j, k}\epsilon_{j, t-k}$ with $\epsilon_{j, t} \sim \mathrm{NID}(0, 1)$ $\forall j=1,...,r$. The model nests the fractional factor models of section \ref{Chapter:2} for $d_1=...=d_r=d$. In matrix form, equation \eqref{eq:1a} is written as
$
\vy = \vf \bm{\Lambda}' + \ve,
$
where $\vy = (\vy_1, ..., \vy_T)'$ is $T \times N$, $\vf = (\vf_1, ..., \vf_T)'$ is $T \times r$ and $\ve = (\ve_1,...,\ve_T)'$ is $T \times N$. In the following, we define $\norm{\mX}= \sqrt{\mathrm{tr}(\mX' \mX)}$. $M < \infty$ is positive and constant.
To extend the proofs of \cite{BaiNg2002}, \cite{Bai2004} to the nonstationary fractional case, we make the following assumptions:
\begin{assumption}[Common stochastic trends]\label{as:cst} 
	The common stochastic trends satisfy the following conditions
	\begin{enumerate}
		\item $\mathrm{E}|z_{j,t}|^{q} \leq M$ for some $q > \mathrm{max}(2, \frac{1}{d - 0.5})$ and for all $t = 1,..., T$, $j=1,...,r$,
		\item The common stochastic trends are mutually independent.
	\end{enumerate}
\end{assumption}
\begin{assumption}[Loadings]\label{as:2}
	The factor loadings $\bm{\Lambda}=(\bm{\lambda_1}', ..., \bm{\lambda_N}')'$ are either deterministic such that $||\bm{\lambda}_i||\leq M \ \forall i=1,...,N$ or stochastic such that $\mathrm{E}||\bm{\lambda}_i||^{4}\leq M \ \forall i=1,...,N$. In either case, $\bm{\Lambda}^{\intercal}\bm{\Lambda}/N \overset{p}{\to} \bm{\Sigma}_N$ as $N \to \infty$, where $\bm{\Sigma}_N$ is a $r \times r$ positive definite deterministic matrix. 
\end{assumption}
\begin{assumption}[Errors]\label{asu:err}
	The errors $\ve_t$  satisfy $\forall i=1,...,N$, $\forall t=1,...,T$
	\begin{enumerate}
		\item $\mathrm{E}[e_{i,t}]=0$, $\mathrm{E}|e_{i,t}|^8 \leq M$,
		\item $\mathrm{E}[\ve_{s}' \ve_{t}/N]=\mathrm{E}[N^{-1}\sum_{i=1}^{N}e_{i,s}e_{i,t}]=\gamma_{N}(s,t)$, $|\gamma_N(s,s)|\leq M \ \forall s=1,...,T$ and \\$T^{-1} \sum_{s=1}^{T} \sum_{t=1}^{T}| \gamma_N(s,t)| \leq M$,
		\item $\mathrm{E}[e_{i,t}e_{j,t}]=\tau_{ij, t}$ with $|\tau_{ij,t}| \leq | \tau_{ij}|$ for some $\tau_{ij}$, and $N^{-1}\sum_{i=1}^{N}\sum_{j=1}^{N}|\tau_{ij}| \leq M$,
		\item $\mathrm{E}[e_{i,t}e_{j,s}]=\tau_{it, js}$ and $N^{-1}T^{-1}\sum_{i=1}^{N}\sum_{j=1}^{N}\sum_{t=1}^{T}\sum_{s=1}^{T}|\tau_{ij, ts}| \leq M$,
		\item For every $(t, s)$, $\mathrm{E}|N^{-1/2}\sum_{i=1}^{N}\{e_{i,s}e_{i,t}-\mathrm{E}[e_{i,s}e_{i,t}]\}|^4 \leq M$.
	\end{enumerate}
\end{assumption}
\begin{assumption}[Independence]\label{asu:ind}
	$\{\bm{\lambda}_i\}$, $\{\vz_t\}$ and $\{\ve_t\}$ are mutually independent stochastic random variables.
\end{assumption}
\noindent
Under these assumptions, corollary \ref{Cor:1} follows directly.
\begin{corollary}\label{Cor:1}
Under assumption \ref{as:cst}, \citet[corollary 2.1]{WuSha2006} yields
\begin{align*}
	T^{-d-0.5} \sum_{t=1}^{\lfloor rT \rfloor} f_{j,t} &\xrightarrow{d} \kappa B_{d}(r), 
\end{align*}
such that $\plim{T \to \infty} T^{-2d} \sum_{t=1}^{T} f_{j,t}^2 \leq M,$
where $\kappa$ is a constant, $r \in [0, 1]$, $B_{d}$ is fractional Brownian motion of type II, $B_{d}(r) = \Gamma(d+1)^{-1}\int_{0}^{r}(r-s)^{d} \mathrm{d}B(s)$ and $B$ is standard Brownian motion generated by $\epsilon_{j,t}$.
\end{corollary}

\noindent
\begin{Theorem}[Consistency of principal components]\label{th:1}
	Suppose assumptions 1 - \ref{asu:ind} hold. Then, for $d \geq 0.5$, the common factors are estimated consistently via principal components up to a rotation such that 
	\begin{align*}
	\delta_{NT}\left(\frac{1}{T}\sum_{t=1}^{T}||\hat{\vf}_t - \mH' \vf_t||^2\right) = O_p(1),
	\end{align*}
	where $\delta_{NT}=\mathrm{min}(N, T^{2d})$, $\hat{\vf}=N^{-1}\vy\vy' \tilde{\vf}T^{-2d}$ is a rescaled version of the principal components together with $\hat{\bm{\Lambda}}=\tilde{\bm{\Lambda}}(\tilde{\bm{\Lambda}}' \tilde{\bm{\Lambda}}N^{-1})^{-1}$, $\tilde{\vf}$ is ${T^{-d}}$ times the eigenvector of $\vy \vy'$, $\tilde{\bm{\Lambda}}'=T^{-2d}\tilde{\vf}' \vy$, and $\bm{H}=(\bm{\Lambda}' \bm{\Lambda}N^{-1})(\vf'\tilde{\vf})T^{-2d}$.
\end{Theorem}

\begin{proof}[Proof of Theorem \ref{th:1}]
	First of all, note that $\tilde{\vf}\tilde{\bm{\Lambda}'}=\hat{\vf}\hat{\bm{\Lambda}}'$ since
	\begin{align*}
	\hat{\vf}\hat{\bm{\Lambda}}'=T^{-2d}\vy\vy' \tilde{\vf} (\tilde{\bm{\Lambda}}'\tilde{\bm{\Lambda}})^{-1}\tilde{\bm{\Lambda}}'=T^{-2d}\tilde{\vf}\tilde{\bm{\Lambda}}' \tilde{\bm{\Lambda}} \tilde{\vf}' \tilde{\vf}(\tilde{\bm{\Lambda}}'\tilde{\bm{\Lambda}})^{-1}\tilde{\bm{\Lambda}}'=\tilde{\vf}\tilde{\bm{\Lambda}}'.
	\end{align*}
	Next 
	\begin{align*}
	\norm{\mH}&\leq \norm{\bm{\Lambda}'\bm{\Lambda}N^{-1}} \norm{\vf'\vf T^{-2d}}^{1/2}\norm{\tilde{\vf}' \tilde{\vf} T^{-2d}}^{1/2},
	\end{align*}
	where the first term is $O_p(1)$ by assumption \ref{as:2}, the second term is $O_p(1)$ by corollary \ref{Cor:1} and the last term is $O_p(1)$ by construction. 
	Now
	\begin{align*}
	\hat{\vf}_t-\mH'\vf_t&=N^{-1}T^{-2d}\left(\sum_{s=1}^{T}\tilde{\vf}_s\vy_s' \vy_t - \tilde{\vf}'\vf \bm{\Lambda}' \bm{\Lambda}{\vf}_t\right)= \\
	&=N^{-1}T^{-2d}\left(\sum_{s=1}^{T}\tilde{\vf}_s\vf_s' \bm{\Lambda}' \ve_t+\sum_{s=1}^{T}\tilde{\vf}_s \ve_s' \bm{\Lambda}\vf_t+\sum_{s=1}^{T}\tilde{\vf_s}\ve_s' \ve_t \right)= \\
	&=T^{-2d}\left( \sum_{s=1}^{T} \tilde{\vf}_s \gamma_N(s,t) + \sum_{s=1}^{T} \tilde{\vf}_s\zeta_{st} + \sum_{s=1}^{T} \tilde{\vf}_s \eta_{st} + \sum_{s=1}^{T} \tilde{\vf}_s\xi_{st} \right),
	\end{align*}
	where, to be consistent with the proofs of \cite{Bai2004}, we define
	\begin{align*}
	&\gamma_{N}(s,t) = N^{-1}\sum_{i=1}^{N}\mathrm{E}(e_{i,s}e_{i,t}), &&\zeta_{s,t} = N^{-1}\sum_{i=1}^{N}(e_{i,s}e_{i,t}-\mathrm{E}(e_{i,s}e_{i,t})), \\
	&\eta_{s,t}=N^{-1}\vf_s'\bm{\Lambda}' \ve_t, &&\xi_{s,t}=N^{-1}\vf_t'\bm{\Lambda}'\ve_s.
	\end{align*}
	Note that
	\begin{align}\label{pr1:1}
	\frac{1}{T}\sum_{t=1}^{T}\norm{\hat{\vf}_t-\mH'\vf_t}^2 &\leq \frac{4}{T}\sum_{t=1}^{T}\Bigg[\norm{T^{-2d}\sum_{s=1}^{T} \tilde{\vf}_s \gamma_N(s,t)}^2 + \norm{T^{-2d}\sum_{s=1}^{T} \tilde{\vf}_s\zeta_{s,t}}^2 \notag \\
	&+ \norm{T^{-2d}\sum_{s=1}^{T} \tilde{\vf}_s \eta_{s,t}}^2 + \norm{T^{-2d} \sum_{s=1}^{T} \tilde{\vf}_s\xi_{s,t}}^2\Bigg].
	\end{align}
	The first argument of equation \eqref{pr1:1} is
	\begin{align*}
	&\frac{1}{T}\sum_{t=1}^{T}\norm{\sum_{s=1}^{T} T^{-2d}\tilde{\vf}_s \gamma_N(s,t)}^2
	\leq T^{-1}\sum_{t=1}^{T}\left[ \sum_{s=1}^{T} \norm{T^{-2d}\tilde{\vf}_s}^2\right] \left[ \sum_{s=1}^{T}\gamma_{N}(s,t)^2 \right],
	\end{align*}
	where from corollary \ref{Cor:1},
	$
	\sum_{s=1}^{T} \norm{T^{-2d}\tilde{\vf}_s}^2=\sum_{s=1}^{T} \mathrm{tr}\left(T^{-4d}\tilde{\vf}_s' \tilde{\vf}_s \right)=
	O_p(T^{-2d}), \label{part1:1}
	$
	and by assumption \ref{asu:err}
	\begin{align}
	&T^{-1}\sum_{s=1}^{T}\gamma_{N}(s,t)^2=T^{-1}\sum_{s=1}^{T}\sum_{t=1}^{T}\gamma_{N}(s,s)\gamma_N(t,t)\rho(s,t)^2\leq \notag \\
	&\leq M T^{-1}\sum_{s=1}^{T}\sum_{t=1}^{T}|\gamma_{N}(s,s)\gamma_N(t,t)|^{1/2}|\rho(s,t)|=M T^{-1}\sum_{s=1}^{T}\sum_{t=1}^{T}|\gamma_N(s,t)|\leq M^2. \label{part1:2}
	\end{align}
	From \eqref{part1:2} it follows that
	$
	T^{-1}\sum_{t=1}^{T}\norm{\sum_{s=1}^{T} T^{-2d}\tilde{\vf}_s \gamma_N(s,t)}^2=O_p(T^{-2d}).
	$
	
	\noindent
	For the second term in \eqref{pr1:1} one has
	\begin{align}
	&\sum_{t=1}^{T}\norm{T^{-2d}\sum_{s=1}^{T} \tilde{\vf}_s\zeta_{s,t}}^2=\sum_{t=1}^{T}\sum_{s=1}^{T}\sum_{u=1}^{T}\left(T^{-4d}\tilde{\vf}_{s}\tilde{\vf}_{u}\right)\zeta_{s,t}\zeta_{u,t}\leq \notag \\
	&\leq \left( \sum_{s=1}^{T}\sum_{u=1}^{T}T^{-4d}(\tilde{\vf}_{s}\tilde{\vf}_{u})^2 \right)^{1/2}\left( T^{-4d} \sum_{s=1}^{T} \sum_{u=1}^{T} \left( \sum_{t=1}^{T} \zeta_{s,t} \zeta_{u,t} \right)^2 \right)^{1/2}\leq \notag \\
	&\leq \left(T^{-2d} \sum_{s=1}^{T} \norm{\tilde{\vf}_s}^2 \right)\left( T^{-4d} \sum_{s=1}^{T} \sum_{u=1}^{T} \left( \sum_{t=1}^{T} \zeta_{s,t} \zeta_{u,t} \right)^2 \right)^{1/2}, \label{p2:1}
	\end{align}
	where $T^{-2d} \sum_{s=1}^{T} \norm{\tilde{\vf}_s}^2$ is $O_p(1)$ from corollary \ref{Cor:1} and
	$\mathrm{E}[(\sum_{t=1}^{T} \zeta_{s,t} \zeta_{u,t})^2]\leq T^2\underset{s, t}{\mathrm{max} \ }\mathrm{E}|\zeta_{s,t}|^4$.
	For the latter term one has by assumption \ref{asu:err}
	\begin{align*}
	\mathrm{E}|\zeta_{s,t}|^4=\frac{1}{N^2}\mathrm{E}\left| N^{-1/2}\sum_{i=1}^{N}\left[e_{i,s}e_{i,t}-\mathrm{E}(e_{i,s}e_{i,t})\right]\right|^4 = N^{-2}M,
	\end{align*}
	as in \cite{BaiNg2002} and, therefore, $\mathrm{E}[(\sum_{t=1}^{T} \zeta_{s,t} \zeta_{u,t})^2]$ is $O_p(\frac{T^2}{N^2})$. Together with \eqref{p2:1} this implies
	\begin{align*}
	&\frac{1}{T}\sum_{t=1}^{T}\norm{T^{-2d}\sum_{s=1}^{T} \tilde{\vf}_s\zeta_{st}}^2\leq T^{-1}O_p(1)O_p(T^{2(1-d)}N^{-1})=O_p(T^{1-2d}N^{-1}),
	\end{align*}
	where $1-2d \leq 0$.
	Considering the third term of equation \eqref{pr1:1} we have
	\begin{align*}
	&\frac{1}{T}\sum_{t=1}^{T}\norm{T^{-2d}\sum_{s=1}^{T} \tilde{\vf}_s \eta_{s,t}}^2 =T^{-1}\sum_{t=1}^{T}\norm{\sum_{s=1}^{T} T^{-2d}\tilde{\vf}_s\vf_s' \bm{\Lambda}' \frac{\ve_t}{N}}^2\leq \\ 
	&\leq T^{-1}\sum_{t=1}^{T}N^{-2}\norm{\ve_t' \bm{\Lambda}}^2 \left(\sum_{s=1}^{T}\norm{T^{-2d}\tilde{\vf}_s\vf_s'}^2\right)\leq \\
	&\leq T^{-1}\sum_{t=1}^{T}N^{-2}\norm{\ve_t' \bm{\Lambda}}^2\left(\sum_{s=1}^{T}\norm{T^{-2d}\vf_s}^2\right)\left(\sum_{s=1}^{T}\norm{\tilde{\vf}_s}^2\right)= \\
	&= T^{-1}\sum_{t=1}^{T}N^{-2}\norm{\ve_t' \bm{\Lambda}}^2\left(\sum_{s=1}^{T}\sum_{l=1}^{r}T^{-2d}f_{l,s}^2\right)\left(T^{-2d}\sum_{s=1}^{T}\norm{\tilde{\vf}_s}^2\right)= \\
	&=  T^{-1}\sum_{t=1}^{T}N^{-2}\norm{\ve_t' \bm{\Lambda}}^2 O_p(1),
	\end{align*}
where the last step follows from corollary \ref{Cor:1}.
	Note that
	\begin{align*}
	\mathrm{E}\left[ \norm{\frac{1}{\sqrt{N}}\sum_{i=1}^{N}e_{i,t}\bm{\lambda}_i}^2 \right]=\frac{1}{N}\sum_{i=1}^{N}\sum_{j=1}^{N}\mathrm{E}(e_{i,t}e_{j,t}\bm{\lambda}_i'\bm{\lambda}_j )\leq \bar{\lambda}^2\frac{1}{N}\sum_{i=1}^{N}\sum_{j=1}^{N}|\tau_{ij}|=\bar{\lambda}^2M,
	\end{align*} with $\bar{\lambda}^2 < \infty$ from assumption \ref{as:2}, and therefore one has for the third term
	\begin{align*}
	&\frac{1}{T}\sum_{t=1}^{T}\norm{T^{-2d}\sum_{s=1}^{T} \tilde{\vf}_s \eta_{s,t}}^2 =T^{-1}\sum_{t=1}^{T}O_p(N^{-1})=O_p(N^{-1}).
	\end{align*}
	Finally, for the last term, one has 
	\begin{align*}
	&\frac{1}{T} \sum_{t=1}^{T}\norm{T^{-2d} \sum_{s=1}^{T} \tilde{\vf}_s\xi_{s,t}}^2=\frac{1}{T}\sum_{t=1}^{T}\norm{\sum_{s=1}^{T}T^{-2d}\tilde{\vf}_s\vf_t' \bm{\Lambda}' \frac{\ve_s}{N}}^2 \leq \\
	&\leq \frac{1}{T} \left( \sum_{s=1}^{T} \norm{N^{-1}\ve_s'\bm{\Lambda}}^2\right)\sum_{t=1}^{T}\sum_{s=1}^{T}\norm{T^{-2d}\tilde{\vf}_s \tilde{\vf}_t'}^2 \leq \\
	&\leq \frac{1}{T} \left( \sum_{s=1}^{T} \norm{N^{-1}\ve_s'\bm{\Lambda}}^2\right) \left(\sum_{t=1}^{T}\norm{T^{-2d}\vf_t}^2\right)\left(\sum_{s=1}^{T}\norm{\tilde{\vf}_s}^2\right)= \\
	&= \frac{1}{T} \left( \sum_{s=1}^{T} \norm{N^{-1}\ve_s'\bm{\Lambda}}^2\right) \left(\sum_{t=1}^{T}\sum_{l=1}^{r}T^{-2d}f_{l,t}^2\right) \left(T^{-2d}\sum_{s=1}^{T}\norm{\tilde{\vf}_s}^2\right)= O_p(N^{-1}).
	\end{align*}
by corollary \ref{Cor:1}. Consequently, theorem \ref{th:1} holds.
\end{proof}

\newpage
\begin{spacing}{1.2}
	\bibliographystyle{dcu}
	\bibliography{literatur.bib}
\end{spacing}
\newpage

\end{document}